\documentclass[12pt,authoryear]{elsarticle}
\journal{}

\makeatletter
\def\ps@pprintTitle{%
 \let\@oddhead\@empty
 \let\@evenhead\@empty
 \def\@oddfoot{\hfill\thepage}%
 \let\@evenfoot\@oddfoot}
\makeatother

\usepackage{amsmath,amssymb,amsfonts,amsthm,graphicx}
\usepackage[bookmarksnumbered,colorlinks=true]{hyperref}
\usepackage[labelfont=bf]{caption}
\usepackage{natbib}

\providecommand{\doi}[1]{\href{https://doi.org/#1}{doi:#1}}
\usepackage{xurl} 
\renewcommand{\doi}[1]{%
 \href{https://doi.org/#1}{\nolinkurl{doi:#1}}%
}

\usepackage{appendix} 
\usepackage{enumitem} 
\usepackage{dsfont} 
\usepackage{xcolor} 
\usepackage{booktabs} 
\usepackage{float} 
\usepackage[table]{xcolor} 
\usepackage{array} 
\newcolumntype{C}[1]{>{\centering\arraybackslash}p{#1}} 
\usepackage{geometry} 
\numberwithin{equation}{section}

\theoremstyle{plain}
\newtheorem{theorem}{Theorem}[section]
\newtheorem{proposition}[theorem]{Proposition}
\newtheorem{lemma}[theorem]{Lemma}

\theoremstyle{definition}
\newtheorem{remark}[theorem]{Remark}

\newcommand{\N}{\mathbb{N}}
\newcommand{\R}{\mathbb{R}}
\newcommand{\PP}{\mathsf{P}} 
\newcommand{\EE}{\mathsf{E}} 
\newcommand{\Bias}{\mathsf{Bias}} 
\newcommand{\Var}{\mathsf{Var}} 
\newcommand{\Cov}{\mathsf{Cov}} 
\newcommand{\bb}[1]{\boldsymbol{#1}}
\newcommand{\OO}{\mathcal{O}}
\newcommand{\oo}{\mathrm{o}}
\newcommand{\rd}{\mathrm{d}}
\newcommand{\ind}{\mathds{1}}
\newcommand{\e}{\varepsilon}
\newcommand{\tr}{\mathrm{tr}}
\newcommand{\etr}{\mathrm{etr}}
\newcommand{\vecc}{\mathrm{vec}}
\newcommand{\vecp}{\mathrm{vecp}}
\newcommand{\argmin}{\operatorname{argmin}}
\newcommand{\argmax}{\operatorname{argmax}}

\allowdisplaybreaks

\begin{document}

\begin{frontmatter}

\title{Wishart kernel density estimation for strongly mixing time series \\on the cone of positive definite matrices}

\author[a1]{L\'eo Belzile}
\author[a2]{Christian Genest}
\author[a3]{Fr\'ed\'eric Ouimet\corref{mycorrespondingauthor}}
\author[a4]{Donald Richards}

\address[a1]{HEC Montr\'eal, Montr\'eal, QC, Canada}
\address[a2]{McGill University, Montr\'eal, QC, Canada}
\address[a3]{Universit\'e du Qu\'ebec \`a Trois-Rivi\`eres, Trois-Rivi\`eres, QC, Canada}
\address[a4]{Penn State University, University Park, PA, USA\vspace{-5mm}}

\cortext[mycorrespondingauthor]{Corresponding author. Email address: frederic.ouimet2@uqtr.ca}

\begin{abstract}
A Wishart kernel density estimator (KDE) is introduced for density estimation in the cone of positive definite matrices. The estimator is boundary-aware and mitigates the boundary bias suffered by conventional KDEs, while remaining simple to implement. Its mean squared error, uniform strong consistency on expanding compact sets, and asymptotic normality are established under the Lebesgue measure and suitable mixing conditions. This work represents the first study of density estimation for dependent data on this space under any metric. For independent observations, an asymptotic upper bound on the mean absolute error is also derived. A simulation study compares the performance of the Wishart KDE with that of the log-Gaussian KDE, another boundary-aware estimator based on the matrix-variate lognormal distribution proposed by Schwartzman [\textit{Int. Stat. Rev.}, 2016, 84(3), 456--486], and with the naive Gaussian KDE on the ambient Euclidean space. When estimating the stationary marginal density of a Wishart autoregressive process for several autoregressive coefficient matrices and innovation covariance matrices, the Wishart KDE exhibits the best overall accuracy and stability. The practical utility of the Wishart KDE is illustrated by estimating the marginal density of a one-year time series of realized covariance matrices computed from 5-minute intra-day returns on Amazon Corp.\ shares and on the Standard \& Poor's 500 exchange-traded fund. All code is publicly available via the \textsf{R} package \texttt{ksm} to facilitate implementation of the method and reproducibility of the findings.
\end{abstract}

\begin{keyword} 
Asymmetric kernel, boundary bias, density estimation, lognormal distribution, realized covariance, smoothing, strong mixing, time series, Wishart autoregressive process, Wishart kernel
\MSC[2020]{Primary: 62G07 Secondary: 60F05, 62E20, 62G05, 62G20, 62H12, 62P05}
\end{keyword}

\end{frontmatter}

\thispagestyle{empty}

\section{Introduction}\label{sec:intro}

In modern statistical applications, many data objects are naturally represented as (symmetric) positive definite matrices. For instance, in diffusion tensor imaging (DTI), each voxel of a brain scan is associated with a $3\times 3$ positive definite matrix that models the local diffusion of water molecules within biological tissue \citep{doi:10.1016/j.nec.2010.12.004}. These matrices provide critical insights into the anisotropic diffusion of water, revealing structural information about white-matter tracts in the brain and enabling the study of neurological disorders, connectivity, and tissue microstructure. In astrophysics, measurements of the Stokes parameters for polarized light are accompanied by a positive definite covariance matrix that quantifies uncertainty and correlations \citep{doi:10.1051/0004-6361/201322271}, a structure essential for understanding the physical properties of light-emitting sources such as stars and galaxies. Likewise, in financial econometrics, realized covariance matrices track the joint variability among multiple assets based on intraday returns, providing a key input for portfolio optimization, risk management, and the modeling of market dynamics \citep{MR2328413}. Beyond these fields, positive definite matrices also arise in radar signal processing and microwave engineering \citep{MR3100414}, and various other areas where covariance-like structures are pivotal to understanding multivariate dependencies.

Despite their prevalence, random positive definite matrices pose unique challenges for analysis compared to random vectors, as their structure depends heavily on the perspective one adopts. Two equally rich and important points of view are briefly described below.

On the one hand, the set of all $d\times d$ real symmetric positive definite matrices,
\[
\mathcal{S}_{++}^d = \left\{M\in \mathcal{S}^d : \forall_{\bb{x} \neq \bb{0}_d} \; \bb{x}^{\top} M \bb{x} > 0 \right\},
\]
is a convex cone in the vector space $\mathcal{S}^d$ of $d\times d$ real symmetric matrices. This perspective arises naturally when equipping the latter with the \emph{Euclidean metric}, defined, for every $U,V\in \mathcal{S}^d$, via the Frobenius norm, viz.
\[
d_E(U,V) = \|U - V\|_F = \sqrt{\tr\{(U - V)^2\}}.
\]
Under this metric, the set $\mathcal{S}_{++}^d$, although not a vector space, retains certain linear-like properties of the surrounding Hilbert space $(\mathcal{S}^d, \|\cdot\|_F) \cong (\R^{d(d + 1)/2}, \|\cdot\|_2)$, making it computationally straightforward to adapt classical statistical methods. Moreover, this view emphasizes the additive structure of positive definite matrices, which is particularly useful in applications where positive definite matrices arise from summations (e.g., realized covariance matrices in finance). However, this perspective largely ignores the inherent curvature of the space and can lead to inaccuracies in applications where the assumed additive structure does not align with the observed data.

\smallskip
On the other hand, the cone $\mathcal{S}_{++}^d$ can be viewed as a Riemannian manifold, reflecting the intrinsic geometry of positive definite matrices. A manifold is a topological space which is locally homeomorphic to Euclidean space, with a Riemannian manifold distinguished by a smoothly varying inner product on its tangent spaces. The tangent space at each point $M\in \mathcal{S}_{++}^d$ is $\mathcal{S}^d$, and a commonly chosen inner product is $\langle U,V \rangle_M = \tr(M^{-1} U M^{-1} V)$ for $U,V\in \mathcal{S}^d$. The corresponding {\it affine-invariant metric} is defined, for every $M,N\in \mathcal{S}_{++}^d$, by
\[
d_A(M,N) = \|\log(M^{-1/2} N M^{-1/2})\|_F,
\]
where `$\log$' is the matrix logarithm and $M^{1/2}$ is the symmetric square root of $M$. This metric makes $\mathcal{S}_{++}^d$ a non-compact Riemannian symmetric space with non-positive curvature. It also admits a closed-form geodesic: the unique geodesic joining $M,N\in \mathcal{S}_{++}^d$ is given, for every $ t\in [0,1]$, by
\[
\gamma(t) = M^{1/2} \exp\{t \, \log(M^{-1/2} N M^{-1/2})\} M^{1/2}.
\]
This makes certain computations intrinsic to $\mathcal{S}_{++}^d$, such as interpolation and averaging. For example, the geodesic midpoint $M \# N = M^{1/2} (M^{-1/2} N M^{-1/2})^{1/2} M^{1/2}$ stays in the cone and serves as the affine-invariant geometric mean \citep{MR3580425}.

The present paper addresses the problem of density estimation for dependent positive definite matrix data under the Euclidean metric, as detailed in Section~\ref{sec:density.estimation}. Before describing the setting and outlining the main contributions, the rationale for choosing the Euclidean metric over non-Euclidean alternatives is explained in Section~\ref{sec:E.vs.NE} through two key application areas: DTI (Section~\ref{sec:E.vs.NE.DTI}) and time series of realized covariance matrices in finance (Section~\ref{sec:E.vs.NE.finance}).

\subsection{Euclidean versus non-Euclidean metrics}\label{sec:E.vs.NE}

\subsubsection{For DTI data}\label{sec:E.vs.NE.DTI}

The Euclidean metric originally was the standard choice for analyzing DTI data, primarily due to its simplicity and compatibility with classical statistical methods. However, the cone $\mathcal{S}_{++}^3$, in which diffusion tensors reside, has an intrinsic non-Euclidean geometry, so the affine-invariant metric $d_A$ was proposed \citep{inria-00070756,inria-00070755,MR2283616,doi:10.1007/s11263-005-3222-z,doi:10.1016/j.sigpro.2005.12.018}. It ensures unique geodesics that remain in $\mathcal{S}_{++}^d$, linear evolution of the log-determinants along geodesics ($|\gamma(t)| = |M|^{1-t} |N|^t, ~t\in [0,1]$), and fundamental geometric properties such as invariance under congruence transformations $X\mapsto G X G^{\top}$ for any invertible $G$ (covering rotations, scalings, and shears), and invariance under the inversion map $X\mapsto X^{-1}$ \citep{MR3580425}. In contrast, the Euclidean metric lacks these invariances; although convex combinations (e.g., simple averaging or convex-weight interpolation) preserve positive definiteness because $\mathcal{S}_{++}^d$ is convex, operations that involve negative weights or extrapolation (e.g., certain filtering and some resampling schemes in registration) can leave the cone and produce tensors that are no longer positive definite \citep{inria-00070756,inria-00070755,doi:10.1007/s11263-005-3222-z,doi:10.1016/j.sigpro.2005.12.018}. Despite these theoretical advantages, the affine-invariant approach is computationally more demanding, as it requires geodesic computations involving matrix logarithms and exponentials.

This challenge motivated the development of the log-Euclidean metric \citep{doi:10.1002/mrm.20965}, which simplifies computations by mapping $\mathcal{S}_{++}^3$ to the vector space $\mathcal{S}^3$ via the matrix logarithm and then applying Euclidean operations there. The log-Euclidean metric retains many of the affine-invariant metric's desirable properties while being significantly less computationally intensive \citep{MR3580425}. Nevertheless, it still entails matrix logarithms and exponentials and is typically more costly computationally than working with the Euclidean metric, which avoids them altogether.

Recent insights and simulations have rekindled interest in the Euclidean metric for certain DTI analyses; see, e.g., \citet{doi:10.1002/mrm.21229,doi:10.1016/j.neuroimage.2009.10.071,doi:10.1007/978-3-642-27343-8_17} or \cite{doi:10.1007/978-3-319-24553-9_20,doi:10.1016/j.neuroimage.2016.05.040}. For instance, thermal acquisition noise often yields approximately normal, rather than log-normal, distributions for apparent diffusion coefficients and their linear combinations \citep{doi:10.1016/j.neuroimage.2009.10.071}, which aligns naturally with arithmetic averaging under the Euclidean metric. Moreover, trace preservation under the Euclidean metric is sometimes argued to have stronger physical justification than the determinant-preserving emphasis of the log-Euclidean and affine-invariant metrics, and simulations show that this trace-based view can reduce noise-induced bias \citep{doi:10.1007/978-3-642-27343-8_17}. Although the log-Euclidean and affine-invariant metrics mitigate the ``swelling effect'' of the Euclidean metric, where averaging can inflate determinants, smoothing or pre-segmentation often reduces such artifacts in practice \citep{doi:10.1016/j.neuroimage.2009.10.071}. Consequently, even though there is no consensus on metric selection \citep{doi:10.1016/j.neuroimage.2016.05.040}, the Euclidean metric remains both practical and theoretically well-founded for many DTI settings, especially when the data are close to normally distributed \citep{doi:10.1016/j.neuroimage.2009.10.071}.

\subsubsection{For time series of realized covariance matrices in finance}\label{sec:E.vs.NE.finance}

In nondegenerate settings, realized covariance matrices are positive definite matrices computed from finite sums of outer products of intraday asset returns. These matrices quantify the joint variability of asset returns and naturally retain the additive structure of the underlying covariation process. This makes the Euclidean metric particularly appealing: additions, averages and local perturbations are interpreted directly in the ambient vector space of symmetric matrices. Additionally, the cone $\mathcal{S}_{++}^d$ is convex, ensuring that positive definiteness is preserved when averaging.

The suitability of the Euclidean metric becomes particularly apparent for the strongly mixing time series of realized covariance matrices considered in Appendix~\ref{app:application}. At the level of their high-frequency construction, \citet{doi:10.1111/j.1468-0262.2004.00515.x} showed, in a fixed-span asymptotic regime where the number of intraday returns increases, that realized covariation converges to the integrated covariance matrix and that the properly normalized estimation error has a mixed Gaussian limit. In the multivariate case, the covariance structure of this limit is governed by integrated products of the entries of the spot covariance matrix. Because these limit theorems are formulated after vectorizing the matrix entries and yield Gaussian fluctuations in the ambient linear coordinates, the additive structure emphasized by the Euclidean metric is both practically and theoretically well-founded. This viewpoint is also consistent with recent nonparametric inference on $\mathcal{S}_{++}^d$: \citet{arXiv:2603.13935v1} developed a two-sample test for symmetric positive definite matrix distributions based on the integrated squared difference between two Wishart kernel density estimators (KDEs), with a closed-form overlap representation that avoids numerical integration. In financial applications, such a procedure is naturally suited to comparing the distribution of realized covariance matrices across two time periods or market regimes. A natural next step is therefore to adapt this two-sample discrepancy to change-point analysis for strongly mixing time series of realized covariance matrices, by scanning over candidate split points and testing whether the distribution before and after the split has changed. This extension would provide additional motivation and practical utility for choosing the Euclidean metric in volatility modeling, and its theoretical and computational development is left for future work.

\subsection{Density estimation on the cone \texorpdfstring{$\mathcal{S}_{++}^d$}{S++d}}\label{sec:density.estimation}

\subsubsection{The boundary bias problem}

The cone $\mathcal{S}_{++}^d$ has a boundary consisting of singular nonnegative definite matrices (i.e., symmetric matrices with at least one zero eigenvalue). This boundary poses significant challenges for density estimation due to the spill-over effect of fixed kernels in classical kernel density estimation (KDE) on Euclidean spaces, commonly referred to as the {\it boundary bias problem}. Indeed, standard symmetric, translation-invariant kernels with a fixed bandwidth are not adaptive to the position of the estimation point. Consequently, when the estimation point is near the boundary, a substantial portion of the kernel's mass extends beyond the support of the target density. This spill-over introduces systematic biases in the estimated density values.

For example, if a traditional KDE is constructed on $\mathcal{S}_{++}^d$ using a symmetric matrix-variate normal distribution as the kernel \citep[p.~71]{MR1738933}, the kernel's mass would spill into regions corresponding to nonpositive definite matrices when estimating near the boundary of $\mathcal{S}_{++}^d$. This is analogous to a Gaussian KDE on $(0,\infty)$ assigning mass to negative values when smoothing near zero.

The structural complexity of the boundary, defined by eigenvalue degeneracies, further complicates correction techniques. Standard boundary kernel methods, such as kernel reflection or boundary-specific designs \citep[Section~4.3]{MR3822372}, are not only impractical in $\mathcal{S}_{++}^d$ but also have yet to be implemented.

\subsubsection{Proposed solutions}

To circumvent boundary issues, the most natural approach to density estimation on $\mathcal{S}_{++}^d$ relies on a transformation, whereby a matrix logarithm is applied to the data. The transformed data are then smoothed in $\mathcal{S}^d \cong \R^{d(d + 1)/2}$ using a conventional multivariate Gaussian kernel, and the results are transformed back to $\mathcal{S}_{++}^d$ via the matrix exponential map (with the usual Jacobian adjustment). This method is equivalent to the use of a classical KDE based on the matrix-variate lognormal distribution on $\mathcal{S}_{++}^d$ defined by \citet{MR3580425}, which generalizes the well-known lognormal distribution on $(0,\infty)$. Henceforth, this estimator is called the log-Gaussian KDE.

One of the goals of the present paper is to show that, in some situations, a better alternative approach consists in working directly on $\mathcal{S}_{++}^d$ by selecting a Wishart kernel with parameters that adapt locally to each estimation point. The resulting estimator is called the {\it Wishart KDE}, formally defined below in Section~\ref{sec:definitions}.

\subsubsection{Asymmetric kernel estimators}\label{sec:asymmetric.kernels}

The Wishart KDE is part of the larger class of asymmetric kernel estimators, which are known to address the boundary bias problem of traditional KDEs due to their naturally adaptive kernel shape. Unlike other methods of boundary correction, such as the reflection method \citep{MR797636,MR1097362} or boundary kernels \citep{MR564251,MR816088,MR1130920,MR1422417,MR1649872,MR1752313}, these estimators are particularly simple to implement and can restore the usual interior bias order near the boundary. Moreover, they remain nonnegative over the entire support of the target density, a desirable property that many boundary‑corrected estimators lack.

The first asymmetric kernel estimator was introduced by \citet{doi:10.2307/2347365}, who used a Dirichlet kernel to smooth compositional data on the simplex. \citet{MR1718494} was the first to study the univariate case (the beta kernel estimator) theoretically. The Wishart KDE introduced herein can be viewed as a matrix-variate generalization of the gamma kernel estimator on $(0,\infty)$ proposed by \citet{MR1794247}. Some theoretical properties of the gamma kernel estimator for strongly mixing observations have been established by \citet{MR2454617,MR2756423}; their time-series results provide a primary motivation for the method developed in this article. In the independent and identically distributed (iid) setting, the literature on gamma kernel estimators covers asymptotic properties \citep{MR1794247,MR2179543,MR2206532,MR2454617,MR2568128,MR2756423,MR2595129,MR2801351,MR3304359,doi:10.1080/10485252.2025.2505639,MR4899357}, bias-reduction techniques \citep{MR3843043,MR3384258}, and Bayesian bandwidth selection \citep{MR4516901,MR4415422}. Refer to \citet{MR3821525} or \cite{MR4319409} for surveys of the broader literature on asymmetric kernel estimators. A unifying perspective on parts of this theory is provided by the notion of associated kernels; see, e.g., \citet{MR3760293,Kokonendji_Some_2021}, \cite{MR4859217} or \cite{MR5025557}.

\subsubsection{Overview of the literature for density estimation on \texorpdfstring{$\mathcal{S}_{++}^d$}{S++d}}

This is the first article to address the theoretical foundations of density estimation for dependent positive definite matrix data under any metric. In the iid setting and under the Euclidean metric, the only works identifiable are \citet[Section~3.1]{MR4358612} and the preprint of \citet[Sections~4~and~5]{arXiv:2009.01983v3}, wherein both parties independently introduced a mean-centered version of the Wishart KDE that is numerically underperforming compared to the mode-centered version defined in Section~\ref{sec:definitions}; see \eqref{eq:Wishart.KDE} and \eqref{eq:Wishart.KDE.mode}.

Given the sparsity of the literature on this subject, here is a brief overview of related works pertaining to density estimation under non-Euclidean geometries for iid samples and closely related deconvolution settings. \citet{inria-00632882v1,MR2838725} gave an early general treatment of the deconvolution problem on $\mathcal{S}_{++}^d$, where observations are corrupted by intrinsic errors, with Wishart errors treated as an important special case. Their estimator is based on the Helgason--Fourier transform and its inversion, with a spectral cut-off that reduces, in the no-noise case, to a sinc-type kernel. \citet{MR3012414} later established minimax rates for a related Helgason--Fourier estimator in Wishart mixture density estimation. In the classical density estimation setting (i.e., without errors), \citet{MR4172886,MR4886483} extended this Fourier-analytic framework to Riemannian symmetric spaces of non-compact type, including $\mathcal{S}_{++}^d$ under the affine-invariant metric, building on the deconvolution solution on the Poincaré upper half-plane given by \citet{MR2676895} and on earlier compact-manifold work such as the Fourier expansion estimator of \citet{MR1056339} and the geodesic-kernel estimator of \citet{MR2179289}. Classical density estimation was also considered by \citet{MR3606419}, who developed kernel estimators on spaces of Gaussian distributions and symmetric positive definite matrices under Fisher and $2$-Wasserstein geometries. This tangent-space and volume-correction approach was later generalized to non-compact symmetric spaces by \citet{MR4522875} through log-Gaussian distributions obtained by pushing Euclidean distributions forward by the exponential map. These works are formulated with respect to intrinsic or non-Euclidean geometries, reference measures, and loss functions that differ from those associated with the ambient Euclidean metric on $\mathcal{S}^d$; hence they are not directly comparable to the Euclidean-metric KDEs considered here and are not pursued further.

\subsection{Outline of the present article}\label{sec:outline}

The rest of this paper is organized as follows. Section~\ref{sec:definitions} introduces preliminary definitions and notation. Section~\ref{sec:main.results} presents the main results: under suitable mixing conditions for dependent data, asymptotics for the mean squared error (MSE) are derived together with the MSE-optimal bandwidth $b$; uniform strong consistency is established on growing compact sets; and asymptotic normality is obtained. For iid observations, an asymptotic upper bound on the mean absolute error (MAE) is also given. Proofs of the main results are collected in Section~\ref{sec:proofs}, with auxiliary technical lemmas deferred to Section~\ref{sec:tech.lemmas}.

Appendix~\ref{app:simulations} reports a simulation study comparing the Wishart KDE with the log-Gaussian KDE and with the Gaussian KDE on the ambient Euclidean space, when estimating the stationary marginal density of a Wishart autoregressive (WAR) process. The results show that the Wishart KDE has the best overall accuracy and stability across the selected autoregressive coefficient matrices and innovation covariance matrices, while the Gaussian KDE clearly underperforms as the innovation covariance matrix becomes closer to being singular, in line with the well-known boundary-bias problem for classical symmetric kernels on constrained supports.

Appendix~\ref{app:application} illustrates the applicability of the method in finance by estimating the marginal density of a one-year time series of realized covariance matrices, based on $5$-minute intraday returns, between the price of Amazon Corporation's shares (ticker symbol: AMZN) and the Standard \& Poor's 500 exchange-traded fund (ticker symbol: SPY). For ease of reference, a list of the abbreviations used herein appears in Appendix~\ref{app:abbreviations}. Reproducibility information, including links to all \textsf{R} code \citep{Rlang} and the \texttt{ksm} package implementation, is provided in Appendix~\ref{app:reproducibility}.

\section{Definitions and notation}\label{sec:definitions}

Given any integer $d\in \N = \{1,2,\ldots\}$, let $\mathcal{S}^d$ and $\mathcal{S}_{++}^d$ denote the spaces of $d\times d$ real matrices that are symmetric and (symmetric) positive definite, respectively. Let $\tr(\cdot)$ stand for the trace operator, $\etr(\cdot) = \exp\{\tr(\cdot)\}$ the exponential trace, and $|\cdot|$ the determinant.

The multivariate gamma function $\Gamma_d$ can be equivalently defined, for every $\alpha > (d-1)/2$, by
\begin{equation}\label{eq:multivariate.gamma}
\Gamma_d(\alpha) = \int_{\mathcal{S}_{++}^d} \etr(- X) \, |X|^{\alpha - (d + 1)/2} \, \rd X = \pi^{d(d-1)/4} \prod_{i=1}^d \Gamma\{\alpha - (i - 1)/2\},
\end{equation}
where $ \, \rd X$ denotes the Lebesgue measure on $\mathcal{S}_{++}^d$ \citep[p.~61]{MR652932}, and $\Gamma_1(\cdot)$ reduces to the ordinary gamma function $\Gamma(\cdot)$.

The probability density function of the Wishart distribution with degrees-of-freedom parameter $\nu\in (d-1,\infty)$ and scale matrix $\Sigma\in \mathcal{S}_{++}^d$ is defined, for every $X\in \mathcal{S}_{++}^d$ relative to $ \, \rd X$, by
\begin{equation}\label{eq:Wishart.kernel}
K_{\nu,\Sigma}(X) = \frac{|X|^{\nu/2 - (d + 1)/2} \, \etr(-\Sigma^{-1} X / 2)}{|2 \Sigma|^{\nu/2} \Gamma_d(\nu/2)}.
\end{equation}
If a $d\times d$ random matrix $\mathfrak{X}$ follows this distribution, one writes $\mathfrak{X}\sim \mathrm{Wishart}_d(\nu,\Sigma)$.

Let $\vecp(\cdot)$ denote the vectorization operator that stacks the columns of the upper triangular portion of a $d\times d$ symmetric matrix on top of each other, viz.
\[
\vecp(X) = (X_{11}, X_{12}, X_{22}, \ldots, X_{1d}, \ldots, X_{dd})^{\top}.
\]
For $\mathfrak{X}\sim \mathrm{Wishart}_d(\nu,\Sigma)$, the expectation and covariance matrix of the vector $\vecp(\mathfrak{X})$ are
\[
\EE\{\vecp(\mathfrak{X})\} = \nu \, \vecp(\Sigma), \quad \Var\{\vecp(\mathfrak{X})\} = 2\nu \, B_d^{\top} \Sigma^{\otimes 2} B_d,
\]
where $\otimes$ denotes the Kronecker product, and $B_d$ is the $d^{2} \times \{d(d + 1)/2\}$ transition matrix between $\vecp(\cdot)$ and the usual vectorization operator $\vecc(\cdot)$; see, e.g., \citet[p.~11]{MR1738933}.

Let $\mathfrak{X}_1,\ldots,\mathfrak{X}_n$ be a sequence of possibly dependent $d\times d$ random positive definite matrices, each of which has density $f$ which is supported on $\mathcal{S}_{++}^d$. For example, the sequence could correspond to the first $n$ components of an $\mathcal{S}_{++}^d$-valued stationary stochastic process $(\mathfrak{X}_t)_{t\in \N}$, where the dependence between the components is typically restricted by a strong mixing condition or some variation. The iid case $\mathfrak{X}_1,\ldots,\mathfrak{X}_n$ is included as a special case. The density $f$ is unknown and referred to as the target density. For a given bandwidth parameter $b\in (0,\infty)$, the Wishart KDE for $f$ is defined, for every $S\in \mathcal{S}_{++}^d$, by
\begin{equation}\label{eq:Wishart.KDE}
\hat{f}_{n,b}^{\,\mathrm{W}}(S) = \frac{1}{n} \sum_{t=1}^n K_{\nu(b,d), b S}(\mathfrak{X}_t),
\end{equation}
where $\nu(b,d) = 1/b + d + 1$ for brevity. As mentioned in Section~\ref{sec:asymmetric.kernels}, this estimator is a matrix-variate extension of the gamma KDE on $(0,\infty)$ introduced by \citet{MR1794247}.

The Wishart KDE offers an adaptive approach to density estimation on $\mathcal{S}_{++}^d$, addressing boundary bias by locally adjusting the Wishart kernel's degrees-of-freedom and scale parameters (and hence its shape) at each estimation point $S$. This adaptability prevents the spill-over effect observed with fixed kernels, where mass extends beyond the support near the boundary, and ensures that $\smash{\hat{f}_{n,b}^{\,\mathrm{W}}}$ remains confined to $\mathcal{S}_{++}^d$. The effectiveness of the Wishart kernel in mitigating boundary effects is demonstrated in Remark~\ref{rem:boundary.bias}, which shows that the pointwise bias is uniformly negligible for any bounded region, regardless of proximity to the boundary.

The choice of the degrees-of-freedom and scale parameters for the Wishart kernel is motivated by the following considerations. If $\mathfrak{W}_{b,S}\sim \mathrm{Wishart}_d(1/b + d + 1, b S)$, then
\begin{equation}\label{eq:Wishart.KDE.mode}
\mathsf{Mode}(\mathfrak{W}_{b,S}) = S,
\end{equation}
as proved in Lemma~\ref{lem:Wishart.kernel.global.bound} of Section~\ref{sec:tech.lemmas}, indicating that $X\mapsto K_{\nu(b,d), b S}(X)$ achieves its maximum at the estimation point $X = S$. Additionally, one has
\begin{equation}\label{eq:expectation.covariance}
\begin{aligned}
\EE(\mathfrak{W}_{b,S}) &= S + b (d + 1) S, \\
\EE\{\vecp(\mathfrak{W}_{b,S} - S) \vecp(\mathfrak{W}_{b,S} - S)^{\top}\} &= 2 b B_d^{\top} S^{\otimes 2} B_d + \oo(b B_d^{\top} S^{\otimes 2} B_d).
\end{aligned}
\end{equation}
Thus, as the bandwidth parameter $b$ approaches $0$, the expectation converges to the estimation point~$S$, while the Frobenius norm of the covariance matrix tends to $0$, resulting in the Wishart kernel concentrating increasingly around $S$.

Throughout the paper, the following notational conventions are adopted. The notation $u = \OO(v)$ means that $\limsup |u / v| \leq C < \infty$ as $n\to \infty$ or $b\downarrow 0$, depending on the context. The positive constant $C$ may depend on the target density $f$, the dimension $d$, the strong mixing exponent $\beta$, and any internal constant in a given proof, but on no other variables unless explicitly written as a subscript. A common occurrence is a local dependence of the asymptotics on a given point $S\in \mathcal{S}_{++}^d$, in which case one writes $u = \OO_S(v)$. The alternative notation $u \ll v$ is also used to mean $u = \OO(v)$. If both $u \ll v$ and $u \gg v$ hold, one writes $u \asymp v$. Similarly, the notation $u = \oo(v)$ means that $\lim |u / v| = 0$ as $n\to \infty$ or $b\downarrow 0$. The subscripts specify the parameters on which the convergence rate may depend. The symbol $\rightsquigarrow$ denotes convergence in distribution. The bandwidth parameter $b = b(n)$ is always implicitly a function of the number of observations, except in Section~\ref{sec:tech.lemmas}.

Given that $(\mathcal{S}^d, \|\cdot\|_F)$ is a $d(d + 1)/2$-dimensional Hilbert space, all notions of continuity and differentiability on $\mathcal{S}^d$ or $\mathcal{S}_{++}^d$ are induced from their counterparts on $\R^{d(d + 1)/2}$. For any $V\in \mathcal{S}^d$,~let
\[
\lambda_1(V) \geq \dots \geq \lambda_d(V)
\]
be its eigenvalues placed in descending order. For any $d,k\in \N$, the shorthand notations
\[
r(d) = \frac{d(d + 1)}{2}, \quad [k] = \{1,\ldots,k\},
\]
will be used frequently.

\section{Main results}\label{sec:main.results}

This section is divided into two subsections: results where some form of dependence (mixing condition) is assumed between the observations (Section~\ref{sec:results.mixing}) and results that are specific to the case of iid observations (Section~\ref{sec:results.iid}). The results in Section~\ref{sec:results.mixing} do not exclude the iid case.

For every result in this section, the following common assumption is made:
\begin{enumerate}[label=(A)]\setlength{\itemsep}{0em}
\item The target density $f$ and its first- and second-order partial derivatives are uniformly continuous and bounded on $\mathcal{S}_{++}^d$. \label{ass:A}
\end{enumerate}

\subsection{Results in the mixing case}\label{sec:results.mixing}

The first theorem investigates the asymptotics of the MSE of $\smash{\hat{f}_{n,b}^{\,\mathrm{W}}}$ under the assumption that the local dependence function has uniformly bounded $L^p$ norm and under a notion of dependence slightly weaker than strong mixing, called $2$-$\alpha$-mixing by \citet{MR1640691}. A similar asymptotic result was proved under strong mixing for the (boundary-modified) gamma KDE ($d = 1$) by \citet[Proposition~1]{MR2756423}.

\begin{theorem}[Mean squared error]\label{thm:MSE}\addcontentsline{toc}{subsection}{Theorem~\ref{thm:MSE}}
Suppose that Assumption~\ref{ass:A} holds. Assume that the local dependence function $g_{t,t'} = f_{\mathfrak{X}_t,\mathfrak{X}_{t'}} - f_{\mathfrak{X}_t} f_{\mathfrak{X}_{t'}}$ satisfies
\begin{equation}\label{eq:ass.2}
G_p \equiv \sup_{|t' - t|\geq 1} \|g_{t,t'}\|_p < \infty,
\end{equation}
for some real $p\in (2,\infty)$, where $\smash{\|g_{t,t'}\|_p = \{\int_{\mathcal{S}_{++}^d \times \mathcal{S}_{++}^d} |g_{t,t'}(X,X')|^p \, \rd X \, \rd X'\}^{1/p}}$. Also, assume that the process $(\mathfrak{X}_t)_{t\in \N}$ is $2$-$\alpha$-mixing, i.e., for every integer $k\in \N$,
\begin{equation}\label{eq:ass.3}
\alpha^{(2)}(k) = \sup_{t\in \N} \alpha\{\sigma(\mathfrak{X}_t), \sigma(\mathfrak{X}_{t+k})\} \leq C k^{-\beta},
\end{equation}
for some constants $C\in (0,\infty)$ and $\beta > 2 (p-1)/(p-2)$. Moreover, for every $S\in \mathcal{S}_{++}^d$, define
\begin{equation}\label{eq:g.psi}
g(S) = (d + 1) \nabla f(S)^{\top} \vecp(S) + \nabla^{\otimes 2} f(S)^{\top} \vecc(B_d^{\top} S^{\otimes 2} B_d), \quad
\psi(S) = \frac{|S|^{-(d + 1)/2}}{2^{r(d) + d/2} \pi^{r(d)/2}},
\end{equation}
where $\nabla$ denotes the gradient vector with respect to the elements of the upper triangular portion of $S$ in the same order as $\vecp(\cdot)$, $\nabla^{\otimes 2}$ is the vectorized Hessian, and $\vecc(\cdot)$ is the vectorization operator that converts a matrix into a vector by stacking the columns of the matrix on top of one another. Then, one has, as $n\to \infty$ and $b=b(n)\downarrow 0$,
\[
\begin{aligned}
\mathrm{MSE}\{\hat{f}_{n,b}^{\,\mathrm{W}}(S)\}
&\equiv \EE\big[\{\hat{f}_{n,b}^{\,\mathrm{W}}(S) - f(S)\}^2\big] \\
&= n^{-1} b^{-r(d)/2} \psi(S) f(S) + b^2 g^2(S) + \oo_S\{n^{-1} b^{-r(d)/2}\} + \oo_S(b^2).
\end{aligned}
\]
In particular, if $f(S)g^2(S)\in (0,\infty)$, the asymptotically optimal choice of $b$, with respect to $\mathrm{MSE}$,~is
\[
b_n^{\star}(S) = n^{-2/\{r(d) + 4\}} \left\{\frac{r(d)}{4} \times \frac{\psi(S) f(S)}{g^2(S)}\right\}^{2/\{r(d) + 4\}},
\]
with
\[
\mathrm{MSE}\{\hat{f}_{n,b_n^{\star}}^{\,\mathrm{W}}(S)\}
= n^{-4/\{r(d) + 4\}} \left[\frac{1 + r(d)/4}{\{r(d)/4\}^{r(d)/\{r(d) + 4\}}}\right] \frac{\{\psi(S) f(S)\}^{4 / \{r(d) + 4\}}}{\{g^2(S)\}^{-r(d) / \{r(d) + 4\}}} + \oo_S\big[n^{-4/\{r(d) + 4\}}\big].
\]
\end{theorem}

\begin{remark}[Boundary behavior of the bias]\label{rem:boundary.bias}
Fix $S\in \mathcal{S}_{++}^d$. The proof of Theorem~\ref{thm:MSE} yields the pointwise expansion:
\[
\Bias\{\hat{f}_{n,b}^{\,\mathrm{W}}(S)\} = b \, g(S) + \oo\{b \, \tr(B_d^{\top} S^{\otimes 2} B_d)\}.
\]
This expression clarifies how the geometry of the cone influences the leading bias through the linear factor $\vecp(S)$ and the quadratic factor $B_d^{\top} S^{\otimes 2} B_d$. These factors shrink to zero along paths for which $\|S\|_F\to 0$, but need not shrink along paths approaching nonzero boundary points.

This formula yields several useful observations. First consider a radial approach. If $S=\tau S_0$ with $S_0\in \mathcal{S}_{++}^d$ fixed and $\tau \downarrow 0$, then $g(S)=\OO(\tau)$ and $\tr(B_d^{\top} S^{\otimes 2} B_d)=\OO(\tau^2)$. Hence
\[
\Bias\{\hat{f}_{n,b}^{\,\mathrm{W}}(S)\}= \OO(b \, \tau) + \oo(b \, \tau^2).
\]
In particular, along the canonical choice $\tau\asymp b$, one gets $\Bias\{\hat{f}_{n,b}^{\,\mathrm{W}}(S)\} = \OO(b^2)$, which is smaller than the interior order $\OO(b)$.

Next, consider a directional approach. Suppose $S$ approaches a rank-$(d-r)$ limit ($r\in [d]$) in the sense that
\[
S = V_0 \, \mathrm{diag}(\tau_1,\ldots,\tau_r,\lambda_{r+1},\ldots,\lambda_d) \, V_0^{\top},
\]
for $\tau_j \downarrow 0$, $\lambda_{r+1},\ldots,\lambda_d\in (0,\infty)$ fixed, and a given eigenbasis as columns in $V_0$. Then
\[
\begin{aligned}
|g(S)|
&= \OO(\|S\|_F) = \OO\left(\sum_{i=1}^r \tau_i + \sum_{i=r+1}^d \lambda_i\right) = \OO\left(\sum_{i=1}^r \tau_i\right) + \OO(1), \\
\tr(B_d^{\top} S^{\otimes 2} B_d)
&= \OO\{\tr(S^{\otimes 2})\} = \OO\left\{\left(\sum_{i=1}^r \tau_i + \sum_{i=r+1}^d \lambda_i\right)^2\right\} = \OO\bigg(\sum_{j=1}^r \tau_j^2\bigg) + \OO(1),
\end{aligned}
\]
where the last $\OO(1)$ absorbs the contribution of the non-shrinking block and the cross terms. Hence,
\[
\Bias\{\hat{f}_{n,b}^{\,\mathrm{W}}(S)\} = \OO\left\{b \left(\sum_{i=1}^r \tau_i + \sum_{i=r+1}^d \lambda_i\right)\right\} + \oo\left\{b \left(\sum_{i=1}^r \tau_i + \sum_{i=r+1}^d \lambda_i\right)^2\right\}.
\]
Thus, whenever the collapsing directions dominate ($r = d$), the boundary bias is no worse than in the interior and can be markedly smaller. If at least one eigenvalue of $S$ stays bounded away from~$0$, the leading term reverts to the interior order $\OO(b)$.

Overall, the estimator enjoys a form of boundary adaptivity: along approaches to the boundary for which the Frobenius mass of $S^{\otimes 2}$ vanishes, the leading coefficient $g(S)$ shrinks and the bias recedes accordingly. In contrast, classical symmetric kernels that ignore the positivity constraint typically allocate non-negligible mass outside the support, which yields a boundary bias that does not vanish with the sample size; see \citet[Section~2]{MR1293239}.
\end{remark}

\begin{remark}[Boundary behavior of the variance]\label{rem:boundary.variance}
Although $\smash{\hat{f}_{n,b}^{\,\mathrm{W}}}$ enjoys reduced bias near the boundary, this benefit entails a compensating rise in variance. This phenomenon is a direct consequence of the adaptive mechanism of the Wishart kernel on $\mathcal{S}_{++}^d$. To avoid the spill-over that plagues symmetric kernels, the kernel reshapes according to the evaluation point $S$, allocating all its mass inside the cone. As $S$ drifts toward $\partial\mathcal{S}_{++}^d$ (i.e., as some eigenvalues of $S$ vanish), the kernel becomes increasingly skewed and sharply concentrated in the cone-admissible directions. This necessary adaptation keeps the bias small, but by creating a more ``peaked'' local window it effectively reduces the number of observations that contribute at $S$, thereby increasing sampling variability.

Formally, from the proof of Theorem~\ref{thm:MSE}, one has
\[
\Var\{\hat{f}_{n,b}^{\,\mathrm{W}}(S)\} = n^{-1} b^{-r(d)/2} \psi(S) f(S) + \oo_S(n^{-1} b^{-r(d)/2}).
\]
The factor $\psi(S)$ captures the geometry-induced inflation: as $S$ approaches $\partial\mathcal{S}_{++}^d$ and $|S|\downarrow 0$, one has $\psi(S)\propto |S|^{-(d + 1)/2}\uparrow \infty$, quantifying the increased variability near the boundary. In words, boundary adaptivity reduces bias because the kernel conforms to the support, yet this very conforming shape concentrates the window and shrinks the effective local sample size, which inflates variance. This bias-variance trade-off is typical for asymmetric kernel estimators \citep[see, e.g.,][]{MR1794247,MR2756423,MR4939549} and also appears in related constructions such as Bernstein estimators on constrained supports \citep[see, e.g.,][]{MR2925964,doi:10.1515/stat-2022-0111}.
\end{remark}

The pointwise strong consistency property of the gamma kernel estimator ($d = 1$) in the strong mixing setting was established by \citet[Proposition~3.3]{MR2454617}, though this result did not extend to uniform strong consistency over a growing sequence of compact sets. In the multivariate iid setting, \citet[Theorem~2]{MR2568128} established uniform strong consistency for certain product kernel estimators on fixed compact sets. Theorem~\ref{thm:unif.conv} shows that, under strict stationarity of the process $(\mathfrak{X}_t)_{t\in \N}$ and a geometric strong mixing condition, the Wishart KDE achieves uniform strong consistency over an expanding sequence of compact sets, covering $\mathcal{S}_{++}^d$ in the limit. This result broadens the scope of the findings reported by \citet[Proposition~3.3]{MR2454617} for the unmodified gamma kernel estimator in the one-dimensional case. The buffer regions between the boundary of $\mathcal{S}_{++}^d$ and the expanding compact sets facilitate control over the partial derivatives of the Wishart kernel $K_{\nu(b,d), b S}$ with respect to entries of $S$; see Lemma~\ref{lem:difference.Wishart.kernels} and its proof.

\begin{theorem}[Uniform strong consistency]\label{thm:unif.conv}\addcontentsline{toc}{subsection}{Theorem~\ref{thm:unif.conv}}
Suppose that Assumption~\ref{ass:A} holds and \eqref{eq:ass.2} holds for some real $p\in (2,\infty)$. Next, assume that the process $(\mathfrak{X}_t)_{t\in \N}$ is strictly stationary and geometrically strongly mixing, i.e., for every integer $k\in \N$,
\begin{equation}\label{eq:ass.4}
\alpha(k) = \sup_{t\in \N} \alpha\{\sigma(\mathfrak{X}_s, \, s\leq t), \sigma(\mathfrak{X}_s, \, s\geq t + k)\} \leq C_1 \rho^k,
\end{equation}
for some constants $C_1\in (0,\infty)$ and $\rho\in [0,1)$. Furthermore, choose a constant $\gamma\in (0,\infty)$ and a sequence $\delta_n\downarrow 0$ which satisfies $1 \leq \delta_n^{-1} \ll (\log n)^{(1 + \gamma)/2}$ and
\begin{equation}\label{eq:ass.5}
\sum_{n=1}^{\infty} \sum_{t=1}^n \PP\{\lambda_1(\mathfrak{X}_t) > \delta_n^{-1}\} < \infty.
\end{equation}
Let
\[
\mathcal{S}_{++}^d(\delta_n) = \{S\in \mathcal{S}_{++}^d : \delta_n \leq \lambda_d(S) \leq \lambda_1(S) \leq \delta_n^{-1}\}.
\]
Then, one has, as $n\to \infty$,
\begin{equation}\label{eq:univ.conv.rate}
\sup_{S\in \mathcal{S}_{++}^d(\delta_n)} \big|\hat{f}_{n,b_n}^{\,\mathrm{W}}(S) - f(S)\big| \ll \delta_n^{-r(d)} (\log n)^{(1 + \gamma)/2} n^{-2/\{r(d) + 4\}}, \quad \text{a.s.},
\end{equation}
using an MSE-optimal bandwidth $b_n \asymp n^{-2/\{r(d) + 4\}}$.
\end{theorem}

\begin{remark}
The condition \eqref{eq:ass.5} is quite mild. As shown in Lemma~\ref{lem:eigenvalue.exponential.bounds}, it is satisfied, for example, if the observations are identically distributed as $\mathfrak{X}_1,\ldots,\mathfrak{X}_n\sim \mathrm{Wishart}_d(\nu,\Sigma)$ and $\delta_n^{-1} \asymp (\log n)^{(1 + \gamma)/2}$, for any $\gamma\in (1,\infty)$, $\nu\in (d-1,\infty)$, and $\Sigma\in \mathcal{S}_{++}^d$.
\end{remark}

\begin{remark}
The rate of convergence in \eqref{eq:univ.conv.rate} is optimal up to a logarithmic factor; cf.~\citet[Theorem~2.2]{MR1640691}. Sharpening this rate would require a more refined bandwidth choice of the form $b_n \asymp L_n n^{-2/{r(d) + 4}}$, where $L_n$ denotes a suitable logarithmic factor. Determining the optimal specification of $L_n$ remains an open problem in the present setting.
\end{remark}

The next result shows that under strict stationarity of the process $(\mathfrak{X}_t)_{t\in \N}$, and uniformly bounded supremum norms for the local dependence function and the fourth-order marginal density functions of the process, the Wishart KDE $\hat{f}_{n,b}^{\,\mathrm{W}}(S)$ is asymptotically normal. The proof follows the big/small blocks argument implemented by \citet[Theorem~2.3]{MR1640691} for classical multivariate KDEs and by \citet[Proposition~2]{MR2756423} in the context of the (boundary-modified) gamma KDE for positive time series. The proof reduces to a verification of Lyapunov's condition on the big-block part of the estimator, after a proper normalization.

\begin{theorem}[Asymptotic normality]\label{thm:asymp.norm}\addcontentsline{toc}{subsection}{Theorem~\ref{thm:asymp.norm}}
Suppose that Assumption~\ref{ass:A} holds. Further, assume that the process $(\mathfrak{X}_t)_{t\in \N}$ is strictly stationary, with marginal distributions satisfying
\begin{equation}\label{eq:ass.6}
\sup_{t_1 < t_2} \|f_{\mathfrak{X}_{t_1},\mathfrak{X}_{t_2}} - f_{\mathfrak{X}_{t_1}} f_{\mathfrak{X}_{t_2}}\|_{\infty} <\infty, \quad \sup_{t_1 < t_2 < t_3 < t_4} \|f_{\mathfrak{X}_{t_1},\mathfrak{X}_{t_2},\mathfrak{X}_{t_3},\mathfrak{X}_{t_4}}\|_{\infty} < \infty,
\end{equation}
and strongly mixing, i.e., for every integer $k\in \N$,
\begin{equation}\label{eq:ass.7}
\alpha(k) = \sup_{t\in \N} \alpha\{\sigma(\mathfrak{X}_s, \, s\leq t), \sigma(\mathfrak{X}_s, \, s\geq t + k)\} \leq C_2 k^{-\beta},
\end{equation}
for some constants $C_2\in (0,\infty)$ and $\beta > \beta(d) \equiv \{3 r(d) + 14\} / \{6 r(d) + 8\}$. Let $S\in \mathcal{S}_{++}^d$ be such that $f(S)\in (0,\infty)$. If the MSE-optimal bandwidth $b_n = n^{-2/\{r(d) + 4\}}$ is chosen as $n\to \infty$, then one has
\[
n^{1/2} b_n^{r(d)/4} \frac{\{\hat{f}_{n,b_n}^{\,\mathrm{W}}(S) - f(S) - b_n g(S)\}}{\sqrt{\psi(S) f(S)}} \rightsquigarrow \mathcal{N}(0,1).
\]
\end{theorem}

\begin{remark}
The condition $\beta > \beta(d)$ is not very restrictive. Indeed, note that
\[
\beta(1) = 17/14 \approx 1.21, \quad \beta(2) = 23/26 \approx 0.88, \quad \beta(3) = 8/11 \approx 0.72,
\]
and $\beta(d)\to 1/2$ as $d\to \infty$. In the one-dimensional case ($d = 1$), \citet{MR2756423} imposed the condition $\beta > 3/2$. A careful analysis, carried out in the proof of Theorem~\ref{thm:asymp.norm}, shows that only $\beta > \beta(1) = 17/14$ is required when using the same arguments.
\end{remark}

\subsection{Results specific to the iid case}\label{sec:results.iid}

The proposition below presents an asymptotic expression for the MAE of the Wishart KDE as well as an explicit asymptotic upper bound. An analogous result was established for classical univariate KDEs by \citet{MR955204}, and then extended to the multivariate setting by \citet{MR1118245}. In the context of asymmetric kernels, this result was first proven for the beta kernel estimator on $[0,1]$ by \citet{MR1985506}, and subsequently generalized to higher dimensions for the Dirichlet kernel estimator on the simplex by \citet{MR4319409}.

One key benefit of the MAE over the MSE, as a measure of performance, is its robustness to outliers and extreme deviations. Since the MAE penalizes errors linearly through the absolute difference $\smash{|\hat{f}_{n,b}^{\,\mathrm{W}}(S) - f(S)|}$, large errors do not disproportionately dominate the overall error measure. This contrasts with the MSE, which penalizes errors quadratically and can be excessively influenced by regions where the estimator performs poorly. Moreover, in many applications, the cost associated with an estimation error is proportional to the magnitude of the error rather than its square, making the MAE more aligned with real-world loss functions.

As pointed out by \citet[Section 2.3.2]{MR3329609}, the corresponding integrated $L^1$ risk is a dimensionless quantity and invariant to monotone changes of scale, enhancing its interpretability and practical relevance. The pointwise MAE considered below can be viewed as its local analog. A comprehensive study of kernel smoothing theory from an $L^1$ perspective is provided by \citet{MR780746}.

\begin{proposition}[Mean absolute error]\label{prop:MAE}\addcontentsline{toc}{subsection}{Proposition~\ref{prop:MAE}}
Suppose that Assumption~\ref{ass:A} holds and the observations $\mathfrak{X}_1,\ldots,\mathfrak{X}_n$ are iid. Recall the definition of $g$ and $\psi$ in \eqref{eq:g.psi}. Then, for any $S\in \mathcal{S}_{++}^d$ such that $f(S)\in (0,\infty)$, one has, as $n\to \infty$ and $b=b_n\downarrow 0$,
\begin{equation}\label{eq:MAE}
\begin{aligned}
\mathrm{MAE}\{\hat{f}_{n,b}^{\,\mathrm{W}}(S)\}
&\equiv \EE\big\{\big|\hat{f}_{n,b}^{\,\mathrm{W}}(S) - f(S)\big|\big\} \\
&= \frac{\sqrt{\psi(S) f(S)}}{n^{1/2} b^{r(d)/4}} \, \EE\left\{\left|Z - \frac{n^{1/2} b^{1+r(d)/4} \, g(S)}{\sqrt{\psi(S) f(S)}}\right|\right\} \\
&\quad+ \OO\left\{\frac{n^{-1} b^{-r(d)/2}}{|S|^{(d + 1)/2}}\right\} + \oo\left\{\frac{n^{-1/2} b^{-r(d)/4}}{|S|^{(d + 1)/4}}\right\} + \oo\{b \, \tr(B_d^{\top} S^{\otimes 2} B_d)\}.
\end{aligned}
\end{equation}
If $n^{1/2} b^{r(d)/4}\to \infty$ as $n\to \infty$ (this is the case, e.g., if $b \asymp n^{-2/\{r(d) + 4\}}$), then one has the bound
\begin{equation}\label{eq:MAE.bound}
\mathrm{MAE}\{\hat{f}_{n,b}^{\,\mathrm{W}}(S)\}
\leq \frac{\sqrt{(2/\pi) \psi(S) f(S)}}{n^{1/2} b^{r(d)/4}} + b \, |g(S)| + \oo\left\{\frac{n^{-1/2} b^{-r(d)/4}}{|S|^{(d + 1)/4}}\right\} + \oo\{b \, \tr(B_d^{\top} S^{\otimes 2} B_d)\}.
\end{equation}
\end{proposition}

\section{Proofs}\label{sec:proofs}

\subsection{Proof of Theorem~\ref{thm:MSE}}

Following \citet[p.~44]{MR1640691}, consider the following decomposition:
\begin{equation}\label{eq:MSE.proof.begin}
\begin{aligned}
\EE\big[\{\hat{f}_{n,b}^{\,\mathrm{W}}(S) - f(S)\}^2\big]
&= \big[\EE\{K_{\nu(b,d), b S}(\mathfrak{X}_1)\} - f(S)\big]^2 + n^{-1} \, \Var\{K_{\nu(b,d), b S}(\mathfrak{X}_1)\} \\
&\quad+ \frac{1}{n^2} \sum_{1 \leq |t - t'| \leq n-1} \Cov\{K_{\nu(b,d), b S}(\mathfrak{X}_t), K_{\nu(b,d), b S}(\mathfrak{X}_{t'})\} \\
&\equiv \mbox{(I)}^2 + (\mbox{II)} + \mbox{(III)}.
\end{aligned}
\end{equation}

For the bias term (I), let $\mathfrak{W}_{b,S}\sim \mathrm{Wishart}_d(1/b + d + 1, b S)$ as in Section~\ref{sec:definitions}. By a second-order Taylor expansion, one has
\[
\begin{aligned}
f(\mathfrak{W}_{b,S}) - f(S)
&= \nabla f(S)^{\top} \vecp(\mathfrak{W}_{b,S} - S) + \frac{1}{2} \nabla^{\otimes 2} f(S)^{\top} \vecp(\mathfrak{W}_{b,S} - S)^{\otimes 2} \\
&\quad+ \frac{1}{2} \{\nabla^{\otimes 2} f(\mathfrak{U}_S) - \nabla^{\otimes 2} f(S)\}^{\top} \vecp(\mathfrak{W}_{b,S} - S)^{\otimes 2},
\end{aligned}
\]
for some random matrix $\mathfrak{U}_S\in \mathcal{S}_{++}^d$ on the line segment joining $\mathfrak{W}_{b,S}$ and $S$ in $\mathcal{S}_{++}^d$. Here, the second-order Taylor expansion is applicable because the subspace $\smash{\mathcal{S}_{++}^d}$ is open and convex in $\mathcal{S}^d$. Taking expectations in the last equation, and then using \eqref{eq:g.psi}, one finds that
\[
\begin{aligned}
&\left|\EE\{K_{\nu(b,d), b S}(\mathfrak{X}_1)\} - f(S) - b g(S)\right| \\
&\quad\leq \frac{1}{2} \left|\nabla^{\otimes 2} f(S)^{\top} \left[\EE\{\vecp(\mathfrak{W}_{b,S} - S)^{\otimes 2}\} - 2 b \, \vecc(B_d^{\top} S^{\otimes 2} B_d)\right]\right| \\
&\quad+ \frac{1}{2} \EE\big[|\nabla^{\otimes 2} f(\mathfrak{U}_S) - \nabla^{\otimes 2} f(S)|^{\top} |\vecp(\mathfrak{W}_{b,S} - S)|^{\otimes 2} \ind_{\{\|\vecp(\mathfrak{W}_{b,S} - S)\|_1 \leq \delta\}}\big] \\
&\quad+ \frac{1}{2} \EE\big[|\nabla^{\otimes 2} f(\mathfrak{U}_S) - \nabla^{\otimes 2} f(S)|^{\top} |\vecp(\mathfrak{W}_{b,S} - S)|^{\otimes 2} \ind_{\{\|\vecp(\mathfrak{W}_{b,S} - S)\|_1 > \delta\}}\big] \\[1mm]
&\quad\equiv \Delta_0 + \Delta_1 + \Delta_2,
\end{aligned}
\]
for some $\delta\in (0,\infty)$ to be chosen below. Note that absolute values on vectors are taken component-wise, i.e., $|\bb{v}| = (|v_1|,\ldots,|v_n|)^{\top}$, and $\|\cdot\|_1$ denotes the $\ell^1$ norm.

By \eqref{eq:expectation.covariance} and Assumption~\ref{ass:A}, one has
\[
\Delta_0 = \oo\{b \, \tr(B_d^{\top} S^{\otimes 2} B_d)\}.
\]
By Assumption~\ref{ass:A}, given any scalar $\e \in (0,\infty)$, there exists $\delta = \delta_{\e,d}\in (0,1]$ such that
\[
\|\vecp(S' - S)\|_1 \leq \delta \quad \implies \quad |\nabla^{\otimes 2} f(S') - \nabla^{\otimes 2} f(S)| < \e,
\]
uniformly for $S,S'\in \mathcal{S}_{++}^d$. Hence,
\[
\Delta_1
\leq \frac{\e}{2} \, \EE\big\{\bb{1}_{r(d)^2}^{\top} |\vecp(\mathfrak{W}_{b,S} - S)|^{\otimes 2}\big\}
= \frac{\e}{2} \, \bb{1}_{r(d)^2}^{\top} \EE\big\{|\vecp(\mathfrak{W}_{b,S} - S)|^{\otimes 2}\big\}.
\]
Applying the Cauchy--Schwarz inequality twice, and then \eqref{eq:expectation.covariance}, one obtains
\[
\begin{aligned}
\Delta_1
&\leq \frac{\e}{2} \, \bb{1}_{r(d)^2}^{\top} \big(\big[\EE\big\{|\vecp(\mathfrak{W}_{b,S} - S)|^2\big\}\big]^{1/2}\big)^{\otimes 2}
= \frac{\e}{2} \, \big(\bb{1}_{r(d)}^{\top} \big[\EE\big\{|\vecp(\mathfrak{W}_{b,S} - S)|^2\big\}\big]^{1/2}\big)^2 \\
&\leq \frac{\e}{2} \, r(d) \, \tr\big[\EE\{\vecp(\mathfrak{W}_{b,S} - S) \vecp(\mathfrak{W}_{b,S} - S)^{\top}\}\big] \\
&\ll \e \, b \, \tr(B_d^{\top} S^{\otimes 2} B_d).
\end{aligned}
\]

Next, the term $\Delta_2$ is bounded. Under Assumption~\ref{ass:A}, the second-order partial derivatives~of~$f$ are uniformly bounded on $\mathcal{S}_{++}^d$, so
\[
\Delta_2 \ll \EE\big[\|\vecp(\mathfrak{W}_{b,S} - S)\|_1^2 \, \ind_{\{\|\vecp(\mathfrak{W}_{b,S} - S)\|_1 > \delta\}}\big].
\]
Given that $x^2 \ind_{\{x > \delta\}} \leq \delta^{-2} x^4$ for all $x\in (0,\infty)$, and $\EE(\mathfrak{W}_{b,S}) - S = b (d + 1) S$, one obtains
\[
\begin{aligned}
\Delta_2
&\ll_{\delta} \EE\big\{\|\vecp(\mathfrak{W}_{b,S} - S)\|_1^4\big\} \ll \EE\big\{\|\mathfrak{W}_{b,S} - S\|_F^4\big\} \ll \EE\big\{\|\mathfrak{W}_{b,S} - \EE(\mathfrak{W}_{b,S})\|_F^4\big\} + b^4 \|S\|_F^4.
\end{aligned}
\]
When $\nu \equiv \nu(b,d) = 1/b + (d + 1)\in \N$, one may use the representation
\[
\mathfrak{W}_{b,S} - \EE(\mathfrak{W}_{b,S}) \stackrel{\mathrm{law}}{=} \sum_{i=1}^{\nu} \mathfrak{Y}_i, \quad \mathfrak{Y}_i = (b S)^{1/2} (\bb{Z}_i \bb{Z}_i^{\top} - I_d) (b S)^{1/2}, \quad \EE(\mathfrak{Y}_i) = 0_{d\times d},
\]
where $\bb{Z}_1,\ldots,\bb{Z}_{\nu} \stackrel{\mathrm{iid}}{\sim} \mathcal{N}_d(\bb{0}_d, I_d)$. Using a fourth-moment inequality for sums of iid, mean zero random elements in a Hilbert space \citep[Eq.~(11)]{MR3431705}, followed by the submultiplicativity of the Frobenius norm, one has
\[
\begin{aligned}
\EE\{\|\mathfrak{W}_{b,S} - \EE(\mathfrak{W}_{b,S})\|_F^4\}
&\ll \nu \, \EE(\|\mathfrak{Y}_1\|_F^4) + \nu^2 \{\EE(\|\mathfrak{Y}_1\|_F^2)\}^2 \\
&\ll \nu b^4 \|S\|_F^4 + \nu^2 b^4 \|S\|_F^4 \\
&\ll b^2 \|S\|_F^4.
\end{aligned}
\]
For arbitrary $\nu > d - 1$, the exact same bound remains valid by analytic continuation, because while the derivation above relies on $\nu \in \N$, the moments of the Wishart distribution are purely polynomial in the degrees of freedom. Therefore,
\[
\Delta_2 \ll b^2 \|S\|_F^4 \ll b^2 \, \{\tr(B_d^{\top} S^{\otimes 2} B_d)\}^2 = \oo\{b \, \tr(B_d^{\top} S^{\otimes 2} B_d)\}.
\]
Putting the above estimates for $\Delta_0$, $\Delta_1$, and $\Delta_2$ together, and then letting $\e\downarrow 0$, yields
\begin{equation}\label{eq:I.estimate}
\mbox{(I)} = b g(S) + \oo\{b \, \tr(B_d^{\top} S^{\otimes 2} B_d)\}.
\end{equation}

For the variance term (II), use Assumption~\ref{ass:A} together with the estimate on the $L^2$ norm of the Wishart kernel in Lemma~\ref{lem:L.q.norm.Wishart.kernel} to deduce that
\begin{equation}\label{eq:II.estimate}
\begin{aligned}
\mbox{(II)} &= n^{-1} \, \Var\{K_{\nu(b,d), b S}(\mathfrak{X}_1)\} \\
&= n^{-1} \EE\big\{K_{\nu(b,d), b S}(\mathfrak{X}_1)^2\big\} - n^{-1} \big[\EE\{K_{\nu(b,d), b S}(\mathfrak{X}_1)\}\big]^2 \\
&= n^{-1} b^{-r(d)/2} \frac{|S|^{-(d + 1)/2}}{2^{r(d) + d/2} \pi^{r(d)/2}} \, \{f(S) + \OO(b \|\nabla f\|_{\infty})\} + \OO(n^{-1} \|f\|_{\infty}^2).
\end{aligned}
\end{equation}

Now, consider the covariance term (III). On the one hand, apply H\"older's inequality with $p\in (2,\infty)$ and $q\in (1,2)$ satisfying $1/p + 1/q = 1$, followed by the estimate on the $L^q$ norm of the Wishart kernel in Lemma~\ref{lem:L.q.norm.Wishart.kernel}, to find that, for all $|t - t'|\geq 1$,
\[
\begin{aligned}
|\Cov\{K_{\nu(b,d), b S}(\mathfrak{X}_t), K_{\nu(b,d), b S}(\mathfrak{X}_{t'})\}|
&\leq \int_{\mathcal{S}_{++}^d\times\mathcal{S}_{++}^d} K_{\nu(b,d), b S}(X) K_{\nu(b,d), b S}(X') |g_{t,t'}(X,X')| \, \rd X \, \rd X' \\
&\leq \|K_{\nu(b,d), b S}\|_q^2 \times \sup_{|t - t'|\geq 1} \|g_{t,t'}\|_p \\
&= \frac{b^{-r(d)/p} |S|^{-(d + 1)/p}}{2^{r(d)/p + d/p} q^{r(d)/q} \pi^{r(d)/p}} \, \{1 + \OO_q(b)\} \times G_p,
\end{aligned}
\]
where $G_p \equiv \sup_{|t - t'|\geq 1} \|g_{t,t'}\|_p < \infty$ by \eqref{eq:ass.2}. On the other hand, using Billingsley's inequality \citep[Corollary~1.1]{MR1640691} in conjunction with the upper bound on the supremum norm of the Wishart kernel in Lemma~\ref{lem:Wishart.kernel.global.bound} shows that
\[
\begin{aligned}
|\Cov\{K_{\nu(b,d), b S}(\mathfrak{X}_t), K_{\nu(b,d), b S}(\mathfrak{X}_{t'})\}|
&\leq 4 \, \|K_{\nu(b,d), b S}\|_{\infty}^2 \, \alpha^{(2)}(|t - t'|) \\
&\leq 4 \, \frac{b^{-r(d)} |S|^{-(d + 1)}}{2^{r(d) + d} \pi^{r(d)}} \, \{1 + \OO(b)\} \times \alpha^{(2)}(|t - t'|).
\end{aligned}
\]
Bringing the last two bounds together, and using \eqref{eq:ass.3} to control $\alpha^{(2)}(|t - t'|)$, one finds
\[
\begin{aligned}
|\mbox{(III)}|
&\ll \frac{2}{n} \sum_{k=1}^{n-1} \min\{G_p \, b^{-r(d)/p} |S|^{-(d + 1)/p}, b^{-r(d)} |S|^{-(d + 1)} C k^{-\beta}\} \\
&\ll \frac{2}{n} \left[\sum_{k=1}^{\lfloor b^{-r(d)/(q \beta)}\rfloor} G_p \, b^{-r(d)/p} |S|^{-(d + 1)/p} + b^{-r(d)} |S|^{-(d + 1)} C \sum_{k=\lfloor b^{-r(d)/(q \beta)}\rfloor + 1}^{\infty} k^{-\beta}\right].
\end{aligned}
\]
Given that, for $\beta\in (1,\infty)$ and any diverging sequence $a(n)\to \infty$,
\begin{equation}\label{eq:tail.summation}
\sum_{k=a(n)+1}^{\infty} k^{-\beta} \asymp \int_{a(n)}^{\infty} x^{-\beta} \, \rd x = \frac{a(n)^{-\beta + 1}}{\beta - 1},
\end{equation}
it follows that
\begin{equation}\label{eq:III.estimate}
\begin{aligned}
|\mbox{(III)}|
&\ll_{\beta} n^{-1} \left[b^{-r(d)/(q \beta)} b^{-r(d)/p} |S|^{-(d + 1)/p} + b^{-r(d)} |S|^{-(d + 1)} \{b^{-r(d)/(q \beta)}\}^{-\beta + 1}\right] \\
&\ll n^{-1} b^{-r(d)/2} \times b^{r(d)\{1/2 - 1/p - 1/(q\beta)\}} \{|S|^{-(d + 1)/p} + |S|^{-(d + 1)}\} \\
&= \oo\left[n^{-1} b^{-r(d)/2} \{|S|^{-(d + 1)/p} + |S|^{-(d + 1)}\}\right].
\end{aligned}
\end{equation}
The last equality holds because $1/2 - 1/p - 1/(q\beta) > 0$, this condition being equivalent to the original assumption $\beta > 2 (p-1)/(p-2)$. The condition $p\in (2,\infty)$ ensures that $\beta\in (1,\infty)$, which is necessary to get the asymptotics of the tail summation in \eqref{eq:tail.summation}. Inserting the estimates \eqref{eq:I.estimate}, \eqref{eq:II.estimate}, and \eqref{eq:III.estimate} into \eqref{eq:MSE.proof.begin} yields the desired conclusion.

\subsection{Proof of Theorem~\ref{thm:unif.conv}}

From the estimate \eqref{eq:I.estimate} on the bias term in the proof of Theorem~\ref{thm:MSE}, one has, for all $S\in \mathcal{S}_{++}^d$,
\[
\begin{aligned}
\big|\EE\{\hat{f}_{n,b_n}^{\,\mathrm{W}}(S)\} - f(S)\big|
&\ll b_n \left\{\max_{k\in \{1,\ldots,r(d)\}} \|(\nabla f)_k\|_{\infty}\right\} \sum_{k=1}^{r(d)} |\vecp(S)_k| \\
&\quad+ b_n \left\{\max_{k\in \{1,\ldots,r(d)^2\}} \|(\nabla^{\otimes 2} f)_k\|_{\infty}\right\} \sum_{k=1}^{r(d)^2} |\vecc(B_d^{\top} S^{\otimes 2} B_d)_k|.
\end{aligned}
\]
The first- and second-order partial derivatives of $f$ are bounded under Assumption~\ref{ass:A}, and
\[
\sum_{k=1}^{r(d)} |\vecp(S)_k| \leq \sqrt{r(d)} \, \|S\|_F \leq \sqrt{r(d)} \, \sqrt{d} \, \|S\|_2
\]
while
\[
\sum_{k=1}^{r(d)^2} |\vecc(B_d^{\top} S^{\otimes 2} B_d)_k| \leq r(d) \, \|B_d^{\top} S^{\otimes 2} B_d\|_F \leq r(d) \, \|B_d\|_F^2 \|S\|_F^2 \leq r(d) \, \|B_d\|_F^2 \, d \, \|S\|_2^2.
\]
Therefore, under the restriction $S\in \mathcal{S}_{++}^d(\delta_n)$ with $\delta_n\in (0,1]$, one has
\begin{equation}\label{eq:bias.rate}
\big|\EE\{\hat{f}_{n,b_n}^{\,\mathrm{W}}(S)\} - f(S)\big| \ll b_n \|S\|_2 + b_n \|S\|_2^2 \ll b_n \, \delta_n^{-2}.
\end{equation}

It remains to control the recentered estimator $\smash{\hat{f}_{n,b_n}^{\,\mathrm{W}}(S) - \EE\{\hat{f}_{n,b_n}^{\,\mathrm{W}}(S)\}}$ on $\mathcal{S}_{++}^d(\delta_n)$ using a Voronoi partition with respect to the Frobenius metric.
Fix any Frobenius mesh radius $w_n\in (0,1]$ and select a set of centers $\{S_{n,1},\ldots,S_{n,N_n}\} \subseteq \mathcal{S}_{++}^d(\delta_n)$ that forms a $w_n$-net under $\|\cdot\|_F$, i.e., for every $S\in \mathcal{S}_{++}^d(\delta_n)$, there exists $j\in [N_n]$ with $\|S - S_{n,j}\|_F \leq w_n$. The $j$th Voronoi cell is
\[
R_{n,j} = \big\{S\in \mathcal{S}_{++}^d(\delta_n) : \forall_{\ell \in [N_n]} \; \|S - S_{n,j}\|_F \leq \|S - S_{n,\ell}\|_F \big\}.
\]
By construction, any $S\in R_{n,j}$ satisfies
\begin{equation}\label{eq:width.small.regions.F}
\|S - S_{n,j}\|_F \leq w_n.
\end{equation}
Moreover, given that $\|S\|_F \leq \sqrt{d}\,\|S\|_2 \leq \sqrt{d}\,\delta_n^{-1}$ for all $S\in \mathcal{S}_{++}^d(\delta_n)$, an obvious covering-number bound is
\begin{equation}\label{eq:cover.number.bound}
N_n \ll \left( {\delta_n^{-1}}/ {w_n}\right)^{r(d)}.
\end{equation}
The partition $R_{n,1}, \dots, R_{n,N_n}$ yields
\begin{equation}\label{eq:sup.sup.representation.voronoi}
\sup_{S\in \mathcal{S}_{++}^d(\delta_n)} \big|\hat{f}_{n,b_n}^{\,\mathrm{W}}(S) - \EE\{\hat{f}_{n,b_n}^{\,\mathrm{W}}(S)\}\big|
= \max_{j\in [N_n]} \sup_{S\in R_{n,j}} \big|\hat{f}_{n,b_n}^{\,\mathrm{W}}(S) - \EE\{\hat{f}_{n,b_n}^{\,\mathrm{W}}(S)\}\big|.
\end{equation}
By a union bound, the Borel--Cantelli lemma, and assumption \eqref{eq:ass.5}, the event that
\[
E_n \equiv \bigcap_{t=1}^n \{\lambda_1(\mathfrak{X}_t) \leq \delta_n^{-1}\}
\]
fails infinitely often has probability $0$ because
\[
\sum_{n=1}^{\infty} \PP\left[\bigcup_{t=1}^n \{\lambda_1(\mathfrak{X}_t) > \delta_n^{-1}\}\right]
\leq \sum_{n=1}^{\infty} \sum_{t=1}^n \PP\{\lambda_1(\mathfrak{X}_t) > \delta_n^{-1}\} < \infty.
\]
On the event $E_n$ (for $n$ large enough), Lemma~\ref{lem:difference.Wishart.kernels} and \eqref{eq:width.small.regions.F} give, for every $j\in [N_n]$,
\[
\sup_{S\in R_{n,j}} \big|\hat{f}_{n,b_n}^{\,\mathrm{W}}(S) - \hat{f}_{n,b_n}^{\,\mathrm{W}}(S_{n,j})\big| \ll b_n^{-r(d)/2} \delta_n^{-r(d)} \times b_n^{-1} \delta_n^{-3} w_n
\]
and
\[
\sup_{S\in R_{n,j}} \big|\EE\{\hat{f}_{n,b_n}^{\,\mathrm{W}}(S)\} - \EE\{\hat{f}_{n,b_n}^{\,\mathrm{W}}(S_{n,j})\}\big|
\ll b_n^{-1} \delta_n^{-1} w_n.
\]
Using \eqref{eq:sup.sup.representation.voronoi}, it follows that
\begin{equation}\label{eq:discretization}
\begin{aligned}
\sup_{S\in \mathcal{S}_{++}^d(\delta_n)} \big|\hat{f}_{n,b_n}^{\,\mathrm{W}}(S) - \EE\{\hat{f}_{n,b_n}^{\,\mathrm{W}}(S)\}\big|
&\leq \max_{j\in [N_n]} \big|\hat{f}_{n,b_n}^{\,\mathrm{W}}(S_{n,j}) - \EE\{\hat{f}_{n,b_n}^{\,\mathrm{W}}(S_{n,j})\}\big| \\
&\quad+ \OO\{b_n^{-r(d)/2 - 1} \delta_n^{-r(d)-3} \, w_n\}.
\end{aligned}
\end{equation}
Given some sequence $\e_n\downarrow 0$ to be determined explicitly later in the proof, choose $w_n$ so that
\[
b_n^{-r(d)/2 - 1} \delta_n^{-r(d) - 3} \, w_n \asymp \e_n,
\]
which makes the discretization cost $\OO(\e_n)$ in \eqref{eq:discretization}. In view of \eqref{eq:cover.number.bound}, this choice implies
\[
N_n \ll \left\{{b_n^{-r(d)/2 - 1} \delta_n^{-r(d) - 4}}/{\e_n}\right\}^{r(d)}.
\]

To conclude, it remains to bound the maximum on the right-hand side of \eqref{eq:discretization}. For every $j\in [N_n]$, write
\[
\hat{f}_{n,b_n}^{\,\mathrm{W}}(S_{n,j}) - \EE\{\hat{f}_{n,b_n}^{\,\mathrm{W}}(S_{n,j})\} = \frac{1}{n} \sum_{t=1}^n Y_{n,j,t},
\]
where one defines, for every integer $t \in [n]$,
\[
Y_{n,j,t} = K_{\nu(b_n,d),b_n S_{n,j}}(\mathfrak{X}_t) - \EE\{K_{\nu(b_n,d),b_n S_{n,j}}(\mathfrak{X}_t)\}.
\]
By Lemma~\ref{lem:Wishart.kernel.global.bound}, note that
\begin{equation}\label{eq:bound.Y}
M_n \equiv \max_{j\in [N_n], t\in [n]} |Y_{n,j,t}| \leq \max_{j\in [N_n]} \|K_{\nu(b_n,d),b_n S_{n,j}}\|_{\infty} \ll b_n^{-r(d)/2} \delta_n^{-r(d)}.
\end{equation}
Now, using a union bound, followed by application of Theorem~1.3~(2) of \citet{MR1640691} (which follows from Bradley's lemma and Bernstein's inequality), one finds that, for any integer $q_n\in [1,n/2]$,
\begin{equation}\label{eq:unif.conv.expansion}
\begin{aligned}
&\PP\left[\max_{j\in [N_n]} \big|\hat{f}_{n,b_n}^{\,\mathrm{W}}(S_{n,j}) - \EE\{\hat{f}_{n,b_n}^{\,\mathrm{W}}(S_{n,j})\}\big| > \e_n\right] \\
&\quad\leq \sum_{j=1}^{N_n} \PP\left(\left|\sum_{t=1}^n Y_{n,j,t}\right| > n \e_n\right) \\
&\quad\leq N_n \left[4 \, \exp\left\{-\frac{\e_n^2 \, q_n}{8 v^2(q_n)}\right\} + 22 \left(1 + \frac{4 M_n}{\e_n}\right)^{1/2} q_n \, \alpha\left(\left\lfloor \frac{n}{2q_n}\right\rfloor\right)\right],
\end{aligned}
\end{equation}
where
\begin{equation}\label{eq:def.p.n.v.q.n}
p_n = \frac{n}{2 q_n}, \quad v^2(q_n) = \frac{2}{p_n^2} \sigma^2(q_n) + \frac{M_n \e_n}{2},
\end{equation}
and
\[
\begin{aligned}
\sigma^2(q_n)
&= \max_{\substack{i\in \{0,\ldots,2q_n - 1\} \\ j\in [N_n]}} \EE\Big(\left[(\lfloor i p_n\rfloor + 1 - i p_n) Y_{n,j,\lfloor i p_n\rfloor + 1} + Y_{n,j,\lfloor i p_n\rfloor + 2} + \dots \right. \Big. \\[-4mm]
&\hspace{33mm}\Big. \left. \dots + Y_{n,j,\lfloor (i+1) p_n\rfloor} + \{(i+1) p_n - \lfloor (i+1) p_n\rfloor\} Y_{n,j,\lfloor (i+1) p_n + 1\rfloor}\right]^2\Big).
\end{aligned}
\]
Using the variance estimate in \eqref{eq:II.estimate}, together with the same covariance-splitting argument as in the proof of Theorem~\ref{thm:MSE} based on \eqref{eq:ass.2} and the geometric strong mixing condition \eqref{eq:ass.4}, the covariance contributions in $\sigma^2(q_n)$ are dominated by the variance contributions. Hence,
\begin{equation}\label{eq:sigma.2.q.n}
\sigma^2(q_n) \ll p_n \, b_n^{-r(d)/2} \delta_n^{-r(d)}.
\end{equation}
Combining \eqref{eq:bound.Y}, \eqref{eq:def.p.n.v.q.n}, and \eqref{eq:sigma.2.q.n} then yields
\[
v^2(q_n) \ll n^{-1} q_n \, b_n^{-r(d)/2} \delta_n^{-r(d)} + b_n^{-r(d)/2} \delta_n^{-r(d)} \e_n.
\]
Choosing $q_n = \lfloor (\log n)^{(1 + \gamma)/2} n^{1/2} b_n^{-r(d)/4} \rfloor$ and $\e_n = \delta_n^{-r(d)} (\log n)^{(1 + \gamma)/2} n^{-1/2} b_n^{-r(d)/4}$ for some constant $\gamma\in (0,\infty)$, one finds that
\[
v^2(q_n) \ll b_n^{-r(d)/2} \delta_n^{-r(d)} \e_n, \quad \frac{\e_n^2 \, q_n}{8 v^2(q_n)} \gg (\log n)^{1 + \gamma},
\]
and
\[
\frac{4 M_n}{\e_n} \ll (\log n)^{-(1 + \gamma)/2} n^{1/2} b_n^{-r(d)/4}, \quad \frac{n}{2 q_n} \asymp (\log n)^{-(1 + \gamma)/2} n^{1/2} b_n^{r(d)/4}.
\]
Consequently, the mixing term involving $\alpha$ on the right-hand side of \eqref{eq:unif.conv.expansion} converges to $0$ at least as fast as $\rho\in [0,1)$ to the power $n^{1/\{r(d) + 4\}}$, which is much faster than the exponential term. Thus,
\[
\begin{aligned}
\PP\left[\max_{j\in [N_n]} \big|\hat{f}_{n,b_n}^{\,\mathrm{W}}(S_{n,j}) - \EE\{\hat{f}_{n,b_n}^{\,\mathrm{W}}(S_{n,j})\}\big| > \e_n\right]
&\ll N_n \exp\left\{-\frac{\e_n^2 \, q_n}{8 v^2(q_n)}\right\} \\
&\ll \left\{{b_n^{-r(d)/2 - 1} \delta_n^{-r(d) - 4}}/{\e_n}\right\}^{r(d)} \exp\{-\beta (\log n)^{1 + \gamma}\},
\end{aligned}
\]
for some constant $\beta = \beta(d)\in (0,\infty)$ which is chosen small enough. Using the above expression~for~$\e_n$ together with the fact that $|\log(b_n)| \ll \log(n)$, the factor in front of the exponential is asymptotically bounded by a finite power of $n$, so one finds, for $n$ large enough,
\[
\PP\left[\max_{j\in [N_n]} \big|\hat{f}_{n,b_n}^{\,\mathrm{W}}(S_{n,j}) - \EE\{\hat{f}_{n,b_n}^{\,\mathrm{W}}(S_{n,j})\}\big| > \e_n\right]
\ll \exp\{-(\beta/2) (\log n)^{1 + \gamma}\}.
\]
This last bound is summable in $n$, so the Borel--Cantelli lemma implies
\[
\sup_{S\in \mathcal{S}_{++}^d(\delta_n)} \big|\hat{f}_{n,b_n}^{\,\mathrm{W}}(S) - \EE\{\hat{f}_{n,b_n}^{\,\mathrm{W}}(S)\}\big| \ll \e_n, \quad \text{a.s.}
\]
Moreover, the absolute bias bound in \eqref{eq:bias.rate} is $\OO(\e_n)$ because
\[
b_n \delta_n^{-2} \ll b_n \delta_n^{-r(d)} (\log n)^{(1 + \gamma)/2} \ll \e_n,
\]
by the assumptions $1 \leq \delta_n^{-1} \ll (\log n)^{(1 + \gamma)/2}$ and $b_n \asymp n^{-2/\{r(d) + 4\}}$ in the statement of the theorem. This concludes the proof.

\subsection{Proof of Theorem~\ref{thm:asymp.norm}}

The proof follows that of \citet[Proposition~2]{MR2756423} and \citet[Theorem 2.3]{MR1640691}. The first step is to show that if $\smash{n^{1/2} b_n^{r(d)/4}}\to \infty$ as $n\to \infty$, then
\[
n^{1/2} b_n^{r(d)/4} \frac{\hat{f}_{n,b_n}^{\,\mathrm{W}}(S) - \EE\{\hat{f}_{n,b_n}^{\,\mathrm{W}}(S)\}}{\sqrt{\psi(S) f(S)}} \equiv n^{-1/2} \sum_{t=1}^n Z_{n,t} \rightsquigarrow \mathcal{N}(0,1),
\]
where, for every integer $t \in [n]$,
\[
Z_{n,t} = \frac{K_{\nu(b_n,d),b_n S}(\mathfrak{X}_t) - \EE\{K_{\nu(b_n,d),b_n S}(\mathfrak{X}_t)\}}{\sqrt{b_n^{-r(d)/2} \psi(S) f(S)}}.
\]
Consider the big/small blocks decomposition:
\[
\sum_{t=1}^n Z_{n,t} = \sum_{i=1}^r V_{n,i} + \sum_{i=1}^r V_{n,i}^{\star} + \sum_{t=r(p+q)+1}^n Z_{n,t},
\]
where $V_{n,i} = Z_{n,(i-1)(p+q) + 1} + \dots + Z_{n,i p + (i-1) q}$ and $V_{n,i}^{\star} = Z_{n,i p + (i-1) q + 1} + \dots + Z_{n,i (p+q)}$. Choosing $r \asymp n^a$, $p \asymp \smash{n^{1-a}}$, $q \asymp n^c$, $a\in (0,1)$, and $c\in (0,1-a)$, one finds
\[
n^{-1} \, \Var\left(\sum_{i=1}^r V_{n,i}^{\star}\right) \ll n^{a + c - 1}\to 0, \quad n^{-1} \, \Var\left(\sum_{t=r(p+q)+1}^n Z_{n,t}\right) \ll n^{a-1}\to 0,
\]
which means that $n^{-1/2} \sum_{t=1}^n Z_{n,t}$ has the same asymptotic distribution as $n^{-1/2} \sum_{i=1}^r V_{n,i}$.

As in \citet[pp.~55--56]{MR1640691}, it can be shown that if $c > [a + \{r(d) + 2\} / \{r(d) + 4\}]/(2\beta)$, then the $V_{n,i}$'s are asymptotically independent and that under the assumptions \eqref{eq:ass.6} and \eqref{eq:ass.7}, they can be replaced at a negligible cost by iid random variables, say $W_{n,1}, \dots, W_{n,r}$.

In particular, to verify the asymptotic normality of $\sum_{i=1}^r V_{n,i}$, it suffices to apply Lyapunov's condition to $\sum_{i=1}^r W_{n,i}$. By applying the Cauchy--Schwarz inequality twice, together with the assumption $\sup_{t_1 < t_2 < t_3 < t_4} \|f_{\mathfrak{X}_{t_1},\mathfrak{X}_{t_2},\mathfrak{X}_{t_3},\mathfrak{X}_{t_4}}\|_{\infty} < \infty$ in \eqref{eq:ass.6} and the MSE-optimal bandwidth choice $b_n = \smash{n^{-2/\{r(d) + 4\}}}$, one finds that
\[
\frac{\sum_{i=1}^r \EE(|W_{n,i}|^3)}{\{r \, \Var(W_{n,1})\}^{3/2}} \leq \frac{\sqrt{\sum_{i=1}^r \EE(|W_{n,i}|^4)} \sqrt{\sum_{i=1}^r \EE(|W_{n,i}|^2)}}{\{r \, \Var(W_{n,1})\}^{3/2}} \ll n^{4 / \{r(d) + 4\} - 3a/2}\to 0,
\]
assuming $\beta > \beta(d) \equiv \{3 r(d) + 14\} / \{6 r(d) + 8\}$; see \citet[p.~56]{MR1640691} for details.

Here, $\beta(d)$ is the smallest real for which the restrictions imposed above, i.e.,
\[
a\in (0,1), \quad 1-a > c > \frac{a + \{r(d) + 2\} / \{r(d) + 4\}}{2\beta}, \quad \frac{4}{r(d) + 4} - \frac{3a}{2} < 0,
\]
hold simultaneously for all $\beta > \beta(d)$. The first and third restrictions are equivalent to
\[
a(d) \equiv \frac{8}{3\{r(d) + 4\}} < a < 1.
\]
The possibility of choosing an appropriate constant $c$ in the second restriction is equivalent to
\[
\beta > \frac{a + \{r(d) + 2\} / \{r(d) + 4\}}{2(1 - a)} \equiv \phi(a).
\]
The function $\phi$ is easily shown to be increasing on $(0,1)$, so a necessary and sufficient condition on $\beta$ is $\beta > \phi\{a(d)\} = \beta(d)$.

By \eqref{eq:I.estimate}, used with $b = b_n$, and taking into account the fact that $S$ is fixed, one deduces that
\[
\EE\{\hat{f}_{n,b_n}^{\,\mathrm{W}}(S)\} - f(S) - b_n \, g(S) = \oo\{b_n \, \tr(B_d^{\top} S^{\otimes 2} B_d)\} = \oo_S(b_n).
\]
Therefore, because $b_n = n^{-2/\{r(d) + 4\}}$,
\[
n^{1/2} b_n^{r(d)/4} \frac{\EE\{\hat{f}_{n,b_n}^{\,\mathrm{W}}(S)\} - f(S) - b_n \, g(S)}{\sqrt{\psi(S) f(S)}} = \oo_S\{n^{1/2} b_n^{1+r(d)/4}\} = \oo_S(1).
\]
Consequently,
\[
\begin{aligned}
n^{1/2} b_n^{r(d)/4} \frac{\hat{f}_{n,b_n}^{\,\mathrm{W}}(S) - f(S) - b_n \, g(S)}{\sqrt{\psi(S) f(S)}}
&= n^{1/2} b_n^{r(d)/4} \frac{\hat{f}_{n,b_n}^{\,\mathrm{W}}(S) - \EE\{\hat{f}_{n,b_n}^{\,\mathrm{W}}(S)\}}{\sqrt{\psi(S) f(S)}} + \oo_S(1) \\
&\rightsquigarrow \mathcal{N}(0,1),
\end{aligned}
\]
by Slutsky's theorem. This concludes the proof.

\subsection{Proof of Proposition~\ref{prop:MAE}}

Let $S\in \mathcal{S}_{++}^d$ such that $f(S)\in (0,\infty)$ be given. By Lemma~2 of \citet{MR760686}, if $\xi_1,\dots,\xi_n$ is an iid sequence of random variables with finite absolute third moment, namely $\EE (|\xi_1|^3) < \infty$, then
\[
\sup_{a\in \R} \left|\EE\left\{\Big|\overline{\xi}_n - \EE[\overline{\xi}_n] - a \sqrt{\Var(\overline{\xi}_n)}\Big|\right\} - \sqrt{\Var(\overline{\xi}_n)} \, \EE(|Z - a|)\right|
\ll \frac{\EE\{|\xi_1 - \EE(\xi_1)|^3\}}{n \, \Var(\xi_1)},
\]
where $\overline{\xi}_n = (\xi_1 + \cdots + \xi_n)/n$, $a\in \R$, and $Z\sim \mathcal{N}(0,1)$. By applying this result with
\[
\xi_t = K_{\nu(b,d), b S}(\mathfrak{X}_t), \quad a = -a(S), \quad a(S) \equiv \frac{\Bias\{\hat{f}_{n,b}^{\,\mathrm{W}}(S)\}}{\sqrt{\Var\{\hat{f}_{n,b}^{\,\mathrm{W}}(S)\}}},
\]
and using the symmetry $\EE\{|Z+a(S)|\} = \EE\{|Z-a(S)|\}$, it is shown below that
\begin{equation}\label{eq:Devroye.bound.applied}
\left|\EE\left\{\big|\hat{f}_{n,b}^{\,\mathrm{W}}(S) - f(S)\big|\right\} - \sqrt{\Var\{\hat{f}_{n,b}^{\,\mathrm{W}}(S)\}} \, \EE\left\{|Z - a(S)|\right\}\right| \ll \frac{n^{-1} b^{-r(d)/2}}{|S|^{(d + 1)/2}}.
\end{equation}
Indeed, to get the last inequality, note that, as $b\downarrow 0$,
\[
\frac{\EE\big\{|\xi_1 - \EE(\xi_1)|^3\big\}}{\Var(\xi_1)} \leq \frac{4 \, \big[\EE(\xi_1^3) + \{\EE(\xi_1)\}^3\big]}{\EE(\xi_1^2) - \{\EE(\xi_1)\}^2} \leq \frac{8 \, \EE(\xi_1^3)}{\EE(\xi_1^2) - \{\EE(\xi_1)\}^2} \ll \frac{\|K_{\nu(b,d), b S}\|_3^3}{\|K_{\nu(b,d), b S}\|_2^2},
\]
and using Lemma~\ref{lem:L.q.norm.Wishart.kernel},
\[
\begin{aligned}
\frac{\|K_{\nu(b,d), b S}\|_3^3}{\|K_{\nu(b,d), b S}\|_2^2}
&= \left[\frac{b^{-r(d)} |S|^{-(d + 1)}}{2^{r(d) + d} 3^{r(d)/2} \pi^{r(d)}} \, \{1 + \OO(b)\}\right] \times \left[\frac{b^{-r(d)/2} |S|^{-(d + 1)/2}}{2^{r(d) + d/2} \pi^{r(d)/2}} \, \{1 + \OO(b)\}\right]^{-1} \\
&= \frac{b^{-r(d)/2} |S|^{-(d + 1)/2}}{2^{d/2} 3^{r(d)/2} \pi^{r(d)/2}} \, \{1 + \OO(b)\}.
\end{aligned}
\]

Now, let $w_{n,b}(S) = n^{-1/2} b^{-r(d)/4} \sqrt{\psi(S) f(S)}$, and recall the definitions of $\psi$ and $g$ from \eqref{eq:g.psi}. The triangle inequality and \eqref{eq:Devroye.bound.applied} yield
\[
\begin{aligned}
&\left|\EE\left\{\big|\hat{f}_{n,b}^{\,\mathrm{W}}(S) - f(S)\big|\right\} - w_{n,b}(S) \, \EE\left\{\left|Z - \frac{b \, g(S)}{w_{n,b}(S)}\right|\right\}\right| \\
&\quad\ll \frac{n^{-1} b^{-r(d)/2}}{|S|^{(d + 1)/2}} + \left|\sqrt{\Var\{\hat{f}_{n,b}^{\,\mathrm{W}}(S)\}} \, \EE\{|Z - a(S)|\} - w_{n,b}(S) \, \EE\left\{\left|Z - \frac{b \, g(S)}{w_{n,b}(S)}\right|\right\}\right|.
\end{aligned}
\]
It is shown in Lemma~7 of \citet{MR760686} that, for all $u,w\in (0,\infty)$ and all $v,z\in \R$,
\[
\left|u \, \EE\left(\left|Z - \frac{v}{u}\right|\right) - w \, \EE\left(\left|Z - \frac{z}{w}\right|\right)\right| \leq \sqrt{\frac{2}{\pi}} \, |u - w| + |v - z|.
\]
(Technically, \citet{MR760686} assume that $v,z\in (0,\infty)$, but the map $x\mapsto \EE\{|Z - x|\}$ is symmetric, so the above formulation follows immediately.)
Hence,
\[
\begin{aligned}
&\left|\EE\left\{\big|\hat{f}_{n,b}^{\,\mathrm{W}}(S) - f(S)\big|\right\} - w_{n,b}(S) \, \EE\left\{\left|Z - \frac{b \, g(S)}{w_{n,b}(S)}\right|\right\}\right| \\
&\quad\ll \frac{n^{-1} b^{-r(d)/2}}{|S|^{(d + 1)/2}} + \left|\sqrt{\Var\{\hat{f}_{n,b}^{\,\mathrm{W}}(S)\}} - \frac{\sqrt{\psi(S) f(S)}}{n^{1/2} \, b^{r(d)/4}}\right| + \big|\EE\{\hat{f}_{n,b}^{\,\mathrm{W}}(S)\} - f(S) - b \, g(S)\big| \\
&\quad\equiv \frac{n^{-1} b^{-r(d)/2}}{|S|^{(d + 1)/2}} + \left|\sqrt{\mbox{(II)}} - \frac{\sqrt{\psi(S) f(S)}}{n^{1/2} \, b^{r(d)/4}}\right| + |\mbox{(I)} - b \, g(S)|,
\end{aligned}
\]
where (I) and (II) are defined as in \eqref{eq:MSE.proof.begin}. This proves \eqref{eq:MAE} after applying the estimates on (I) and (II) found in \eqref{eq:I.estimate} and \eqref{eq:II.estimate}, respectively. The bound \eqref{eq:MAE.bound} is a direct consequence of \eqref{eq:MAE} and the trivial bound $\EE(|Z - u|) \leq \smash{\sqrt{2/\pi}} + |u|$. This concludes the proof.

\section{Technical lemmas}\label{sec:tech.lemmas}

The first lemma presents a bound on the supremum norm of the Wishart kernel in \eqref{eq:Wishart.kernel}. For a primer on matrix calculus, see, e.g., \citet{Petersen_Pedersen_2012_online}.

\begin{lemma}\label{lem:Wishart.kernel.global.bound}\addcontentsline{toc}{subsection}{Lemma~\ref{lem:Wishart.kernel.global.bound}}
Let $S\in \mathcal{S}_{++}^d$ be given. Then, as $b\downarrow 0$, one has
\[
\|K_{\nu(b,d), b S}\|_{\infty} \leq \frac{b^{-r(d)/2} |S|^{-(d + 1)/2}}{2^{r(d)/2 + d/2} \pi^{r(d)/2}} \, \{1 + \OO(b)\}.
\]
\end{lemma}

\begin{proof}[\bf Proof of Lemma~\ref{lem:Wishart.kernel.global.bound}]
Given that $|X| I_d = X \mathrm{adj}(X)$ and $X$ is symmetric, one has
\[
\frac{\partial}{\partial X} \, |X| = \mathrm{adj}(X) = |X| X^{-1};
\]
see, e.g., \citet[Eq.~(49)]{Petersen_Pedersen_2012_online}. Furthermore, for all $A\in \mathcal{S}^d$, one has
\[
\frac{\partial}{\partial X} \, \tr(A X) = A;
\]
see, e.g., \citet[Eq.~(100)]{Petersen_Pedersen_2012_online}. Hence, differentiating the expression for the Wishart density in \eqref{eq:Wishart.kernel} and applying the product rule yields
\[
\begin{aligned}
\frac{\partial}{\partial X} \, K_{\nu(b,d), b S}(X)
&= \frac{\{1/(2 b)\} |X|^{1/(2 b) - 1} |X| X^{-1} \etr\{-(b S)^{-1} X / 2\}}{|2 b S|^{1/(2 b) + (d + 1)/2} \Gamma_d\{1/(2 b) + (d + 1)/2\}} \\
&\quad+ \frac{|X|^{1/(2 b)} \etr\{-(b S)^{-1} X / 2\} \{-(b S)^{-1} / 2\}}{|2 b S|^{1/(2 b) + (d + 1)/2} \Gamma_d\{1/(2 b) + (d + 1)/2\}} \\
&= \left[\{1/(2 b)\} X^{-1} -(b S)^{-1} / 2\right] K_{\nu(b,d), b S}(X).
\end{aligned}
\]
Equating this last expression to the zero matrix, $0_{d\times d}$, and solving for $X$ shows that the mode of the $\mathrm{Wishart}_d\{\nu(b,d),b S\}$ distribution is equal to $S$.

Using the second representation of the multivariate gamma function in \eqref{eq:multivariate.gamma}, it follows that
\[
\|K_{\nu(b,d), b S}\|_{\infty}
= K_{\nu(b,d), b S}(S)
= \frac{\{1/(2 b)\}^{d/(2 b) + r(d)} |S|^{-(d + 1)/2} \exp\{-d / (2 b)\}}{\pi^{d(d-1)/4} \prod_{i=1}^d \Gamma\{1/(2 b) + (d - i)/2 + 1\}}.
\]
It is known that $\sqrt{2\pi} e^{-x} x^{x + 1/2} \leq \Gamma(x + 1)$ for every real $x \in (0,\infty)$; see, e.g., \citet[p.~58]{MR1483074}. One deduces that
\[
\begin{aligned}
\|K_{\nu(b,d), b S}\|_{\infty}
&\leq \frac{\{1/(2 b)\}^{d/(2 b) + r(d)} |S|^{-(d + 1)/2} e^{-d / (2 b)}}{\pi^{d(d-1)/4} (2\pi)^{d/2} \prod_{i=1}^d e^{-1/(2 b) - (d-i)/2} \{1/(2 b) + (d-i)/2\}^{1/(2 b) + (d-i+1)/2}} \\
&= \frac{\{1/(2 b)\}^{r(d)/2} |S|^{-(d + 1)/2}}{\pi^{d(d-1)/4} (2\pi)^{d/2} \prod_{i=1}^d e^{-(d-i)/2} \left[1 + \{(d-i)/2\} / \{1/(2 b)\}\right]^{1/(2 b) + (d-i+1)/2}} \\
&= \frac{\{1/(2 b)\}^{r(d)/2} |S|^{-(d + 1)/2}}{\pi^{d(d-1)/4} (2\pi)^{d/2}} \, \{1 + \OO(b)\},
\end{aligned}
\]
where the last equality is a consequence of the fact that $e^{-x} (1 + x/n)^n = 1 + \OO_x(n^{-1})$ as $n\to \infty$, for all $x\in \R$. The conclusion follows.
\end{proof}

The second lemma provides an upper bound on the absolute difference of two Wishart kernels with different scale matrices, pointwise and under expectation. The main trick is an interpolation combined with the fundamental theorem of calculus.

\begin{lemma}\label{lem:difference.Wishart.kernels}\addcontentsline{toc}{subsection}{Lemma~\ref{lem:difference.Wishart.kernels}}
Let $X,S,S'\in \mathcal{S}_{++}^d$ be given. Then, as $b\downarrow 0$, one has
\[
\begin{aligned}
\big|K_{\nu(b,d),b S'}(X) - K_{\nu(b,d), b S}(X)\big|
&\leq \frac{b^{-r(d)/2} \min(|S|,|S'|)^{-(d + 1)/2}}{2^{r(d)/2 + d/2} \pi^{r(d)/2}} \, \{1 + \OO(b)\} \\
&\quad\times \frac{\sqrt{d} \, [\nu(b,d) + b^{-1} \, \min\{\lambda_d(S),\lambda_d(S')\}^{-1} \lambda_1(X)]}{2 \min\{\lambda_d(S),\lambda_d(S')\}} \|S - S'\|_F.
\end{aligned}
\]
Moreover, if a random matrix $\mathfrak{X}$ has a bounded density $f$ supported on $\mathcal{S}_{++}^d$, then
\[
\EE\big\{|K_{\nu(b,d),b S'}(\mathfrak{X}) - K_{\nu(b,d), b S}(\mathfrak{X})|\big\} \leq \frac{\nu(b,d) \, (1 + 2 d) \sqrt{d} \, \|f\|_{\infty}}{2 \min\{\lambda_d(S),\lambda_d(S')\}} \|S - S'\|_F.
\]
\end{lemma}

\begin{proof}[\bf Proof of Lemma~\ref{lem:difference.Wishart.kernels}]
Let $\nu\in (d + 1,\infty)$ and $X,\Sigma,\Sigma'\in \mathcal{S}_{++}^d$ be given. Define $\Sigma_t = t \, \Sigma + (1 - t) \Sigma'$ for every real $t \in [0,1]$. Using \citet[Eq.~(46)]{Petersen_Pedersen_2012_online}, one has
\[
\begin{aligned}
\frac{ \, \rd}{ \, \rd t} |\Sigma_t|^{-\nu/2}
&= - \frac{\nu}{2} |\Sigma_t|^{-\nu/2 -1} \frac{ \, \rd}{ \, \rd t} \, |\Sigma_t| = - \frac{\nu}{2} |\Sigma_t|^{-\nu/2 -1} |\Sigma_t| \, \tr\left(\Sigma_t^{-1} \, \frac{ \, \rd}{ \, \rd t} \, \Sigma_t\right) \\
&= - \frac{1}{2} |\Sigma_t|^{-\nu/2} \, \tr\left\{\nu \Sigma_t^{-1} (\Sigma - \Sigma')\right\}.
\end{aligned}
\]
Also, using \citet[Eqs.~(63),~(137)]{Petersen_Pedersen_2012_online}, one finds that
\[
\frac{ \, \rd}{ \, \rd t} \, \left\{-\frac{1}{2} \tr\left(\Sigma_t^{-1} X\right)\right\}
= \frac{1}{2} \tr\left(\Sigma_t^{-1} X \Sigma_t^{-1} \frac{ \, \rd}{ \, \rd t} \, \Sigma_t\right)
= \frac{1}{2} \tr\left\{\Sigma_t^{-1} X \Sigma_t^{-1} (\Sigma - \Sigma')\right\}.
\]
Combining the last two equations yields
\[
\begin{aligned}
\frac{ \, \rd}{ \, \rd t} K_{\nu,\Sigma_t}(X)
&= \left[- \frac{1}{2} \, \tr\left\{\nu \Sigma_t^{-1} (\Sigma - \Sigma')\right\} + \frac{1}{2} \tr\left\{\Sigma_t^{-1} X \Sigma_t^{-1} (\Sigma - \Sigma')\right\}\right] K_{\nu,\Sigma_t}(X) \\
&= - \frac{1}{2} \tr\left\{(\nu \Sigma_t^{-1} - \Sigma_t^{-1} X \Sigma_t^{-1}) (\Sigma - \Sigma')\right\} K_{\nu,\Sigma_t}(X).
\end{aligned}
\]
By the fundamental theorem of calculus and the triangle inequality for integrals, one deduces that
\begin{equation}\label{eq:lem:differences.of.Wishart.densities.eq.begin}
\begin{aligned}
\big|K_{\nu,\Sigma'}(X) - K_{\nu,\Sigma}(X)\big|
&\leq \int_0^1 \left|\frac{ \, \rd}{ \, \rd t} K_{\nu,\Sigma_t}(X)\right| \, \rd t \\
&= \frac{1}{2} \int_0^1 \left|\tr\left\{(\nu \Sigma_t^{-1} - \Sigma_t^{-1} X \Sigma_t^{-1}) (\Sigma - \Sigma')\right\}\right| K_{\nu,\Sigma_t}(X) \, \rd t.
\end{aligned}
\end{equation}
Considering that control over the supremum norm of $K_{\nu,\Sigma_t}$ has already been established by Lemma~\ref{lem:Wishart.kernel.global.bound}, it remains to bound the trace factor in the integral. By Cauchy--Schwarz for the Frobenius norm, the inequality $\|\cdot\|_F \leq \sqrt{d} \, \|\cdot\|_2$, and the submultiplicativity and triangle inequality for the spectral norm $\|\cdot\|_2$, one has, for every real $t \in [0,1]$,
\[
\begin{aligned}
\left|\tr\left\{(\nu \Sigma_t^{-1} - \Sigma_t^{-1} X \Sigma_t^{-1}) (\Sigma - \Sigma')\right\}\right|
&\leq \|\nu \Sigma_t^{-1} - \Sigma_t^{-1} X \Sigma_t^{-1}\|_F \|\Sigma - \Sigma'\|_F \\
&\leq \sqrt{d} \, \|\nu I_d - \Sigma_t^{-1} X\|_2 \|\Sigma_t^{-1}\|_2 \|\Sigma - \Sigma'\|_F \\
&\leq \sqrt{d} \, (\nu \|I_d\|_2 + \|\Sigma_t^{-1} X\|_2) \|\Sigma_t^{-1}\|_2 \|\Sigma - \Sigma'\|_F.
\end{aligned}
\]
Using Weyl's inequality for the smallest eigenvalue of a sum of two symmetric matrices, $\lambda_d(A + B) \geq \lambda_d(A) + \lambda_d(B)$, one has
\[
\begin{aligned}
\|\Sigma_t^{-1}\|_2^{-1}
&= \lambda_d(\Sigma_t) \geq \lambda_d(t \, \Sigma) + \lambda_d\{(1 - t) \, \Sigma'\} \\
&= t \, \lambda_d(\Sigma) + (1 - t) \, \lambda_d(\Sigma') \\
&\geq \min\{\lambda_d(\Sigma),\lambda_d(\Sigma')\}.
\end{aligned}
\]
It follows from the last two displays that
\begin{equation}\label{eq:lem:differences.of.Wishart.densities.eq.begin.next}
\left|\tr\left\{(\nu \Sigma_t^{-1} - \Sigma_t^{-1} X \Sigma_t^{-1}) (\Sigma - \Sigma')\right\}\right|
\leq \frac{\sqrt{d} \, (\nu + \|\Sigma_t^{-1} X\|_2)}{\min\{\lambda_d(\Sigma),\lambda_d(\Sigma')\}} \|\Sigma - \Sigma'\|_F.
\end{equation}
Moreover, given that
\[
\|\Sigma_t^{-1} X\|_2 \leq \|\Sigma_t^{-1}\|_2 \|X\|_2 \leq \min\{\lambda_d(\Sigma),\lambda_d(\Sigma')\}^{-1} \lambda_1(X),
\]
one deduces from \eqref{eq:lem:differences.of.Wishart.densities.eq.begin} that
\[
\big|K_{\nu,\Sigma'}(X) - K_{\nu,\Sigma}(X)\big|
\leq \sup_{t\in [0,1]} \|K_{\nu,\Sigma_t}\|_{\infty} \times \frac{\sqrt{d} \, [\nu + \min\{\lambda_d(\Sigma),\lambda_d(\Sigma')\}^{-1} \lambda_1(X)]}{2 \min\{\lambda_d(\Sigma),\lambda_d(\Sigma')\}} \|\Sigma - \Sigma'\|_F.
\]
Now let $\nu = \nu(b,d)$, $\Sigma = b S$, and $\Sigma' = b S'$. Using the bound on the Wishart kernel in Lemma~\ref{lem:Wishart.kernel.global.bound}, with the log-determinant concavity on $\mathcal{S}_{++}^d$, $|t S + (1-t) S'| \geq |S|^t |S'|^{1-t} \geq \min(|S|,|S'|)$, one has
\[
\begin{aligned}
\big|K_{\nu(b,d),b S'}(X) - K_{\nu(b,d), b S}(X)\big|
&\leq \frac{b^{-r(d)/2} \min(|S|,|S'|)^{-(d + 1)/2}}{2^{r(d)/2 + d/2} \pi^{r(d)/2}} \, \{1 + \OO(b)\} \\
&\quad\times \frac{\sqrt{d} \, [\nu(b,d) + b^{-1} \min\{\lambda_d(S),\lambda_d(S')\}^{-1} \lambda_1(X)]}{2 \min\{\lambda_d(S),\lambda_d(S')\}} \|S - S'\|_F,
\end{aligned}
\]
which proves the first claim of the lemma.

To prove the second claim, return to \eqref{eq:lem:differences.of.Wishart.densities.eq.begin}. For every scalar $t \in [0,1]$, write
\[
Y_t(X) = \Sigma_t^{-1/2} X \Sigma_t^{-1/2}.
\]
Then, by cyclic symmetry of the trace and the Cauchy--Schwarz inequality for the Frobenius norm,
\[
\begin{aligned}
\left|\tr\left\{(\nu \Sigma_t^{-1} - \Sigma_t^{-1} X \Sigma_t^{-1}) (\Sigma - \Sigma')\right\}\right|
&= \left|\tr\left[\{\nu I_d - Y_t(X)\} \Sigma_t^{-1/2} (\Sigma - \Sigma') \Sigma_t^{-1/2}\right]\right| \\
&\leq \|\nu I_d - Y_t(X)\|_F \|\Sigma_t^{-1/2} (\Sigma - \Sigma') \Sigma_t^{-1/2}\|_F \\
&\leq \frac{\sqrt{d} \, \{\nu + \sqrt{\tr[Y_t(X)^2]}\}}{\min\{\lambda_d(\Sigma),\lambda_d(\Sigma')\}} \|\Sigma - \Sigma'\|_F .
\end{aligned}
\]
Under the assumption that the density of $\mathfrak{X}$ is bounded, i.e., $\|f\|_{\infty} < \infty$, one gets
\[
\begin{aligned}
\EE\big\{|K_{\nu,\Sigma'}(\mathfrak{X}) - K_{\nu,\Sigma}(\mathfrak{X})|\big\}
&\leq \frac{1}{2} \int_0^1 \int_{\mathcal{S}_{++}^d} \left|\tr\left\{(\nu \Sigma_t^{-1} - \Sigma_t^{-1} X \Sigma_t^{-1}) (\Sigma - \Sigma')\right\}\right| K_{\nu,\Sigma_t}(X) f(X) \, \rd X \, \rd t \\
&\leq \frac{\sqrt{d} \, \|f\|_{\infty} \|\Sigma - \Sigma'\|_F}{2 \min\{\lambda_d(\Sigma),\lambda_d(\Sigma')\}}
\int_0^1 \int_{\mathcal{S}_{++}^d} \left[\nu + \sqrt{\tr\{Y_t(X)^2\}}\right] K_{\nu,\Sigma_t}(X) \, \rd X \, \rd t \\
&= \frac{\sqrt{d} \, \|f\|_{\infty} \|\Sigma - \Sigma'\|_F}{2 \min\{\lambda_d(\Sigma),\lambda_d(\Sigma')\}}
\int_0^1 \left[\nu + \EE\left\{\sqrt{\tr(\mathfrak{Y}^2)}\right\}\right] \, \rd t,
\end{aligned}
\]
where $\mathfrak{Y}\sim \mathrm{Wishart}_d(\nu,I_d)$. By applying Jensen's inequality and the second moment formula for $\EE(\mathfrak{Y}^2)$ in \citet[p.~99]{MR1738933}, one obtains
\[
\EE\left\{\sqrt{\tr(\mathfrak{Y}^2)}\right\}
\leq \sqrt{\EE\left\{\tr(\mathfrak{Y}^2)\right\}}
= \sqrt{\tr\left\{\EE(\mathfrak{Y}^2)\right\}}
= \sqrt{\nu (\nu + 1) \tr(I_d^2) + \nu \, \tr(I_d)^2}
\leq 2 \nu d,
\]
which shows that
\[
\EE\big\{|K_{\nu,\Sigma'}(\mathfrak{X}) - K_{\nu,\Sigma}(\mathfrak{X})|\big\} \leq \frac{\nu \, (1 + 2 d) \sqrt{d} \, \|f\|_{\infty}}{2 \min\{\lambda_d(\Sigma),\lambda_d(\Sigma')\}} \|\Sigma - \Sigma'\|_F.
\]
Letting $\nu = \nu(b,d)$, $\Sigma = b S$, and $\Sigma' = b S'$ proves the second claim of the lemma.
\end{proof}

The third lemma studies the asymptotics of the $L^q$ norm of the Wishart kernel for all $q \in (1,\infty)$.

\begin{lemma}\label{lem:L.q.norm.Wishart.kernel}\addcontentsline{toc}{subsection}{Lemma~\ref{lem:L.q.norm.Wishart.kernel}}
Let $p, q\in (1,\infty)$ and $S\in \mathcal{S}_{++}^d$ be given such that $1/p + 1/q = 1$. Then, as $b\downarrow 0$, one has
\[
\|K_{\nu(b,d), b S}\|_q^2 = \frac{b^{-r(d)/p} |S|^{-(d + 1)/p}}{2^{r(d)/p + d/p} q^{r(d)/q} \pi^{r(d)/p}} \, \{1 + \OO_q(b)\}.
\]
In particular, for $p = q = 2$, the result simplifies to
\begin{equation}\label{eq:L.2.norm.Wishart.kernel}
\|K_{\nu(b,d), b S}\|_2^2 = \frac{b^{-r(d)/2} |S|^{-(d + 1)/2}}{2^{r(d) + d/2} \pi^{r(d)/2}} \, \{1 + \OO(b)\}.
\end{equation}
\end{lemma}

\begin{remark}
For an alternative proof of \eqref{eq:L.2.norm.Wishart.kernel} using a local symmetric matrix-variate normal approximation to the Wishart density, see \citet[p.~7]{MR4358612}.
\end{remark}

\begin{proof}[\bf Proof of Lemma~\ref{lem:L.q.norm.Wishart.kernel}]
For all $q\in (1,\infty)$, $b\in (0,1/(d + 1))$, and $S\in \mathcal{S}_{++}^d$, define
\[
\begin{aligned}
A_{b,q}(S)
&= \frac{K_{\nu(b,d), b S}^q(\cdot)}{K_{q/b + d + 1,b S/q}(\cdot)}
= \frac{|2 b \, S/q|^{q/(2 b) + (d + 1)/2} \, \Gamma_d\{q/(2 b) + (d + 1)/2\}}{|2 b \, S|^{q/(2 b) + q(d + 1)/2} \, \Gamma_d^q\{1/(2 b) + (d + 1)/2\}} \\
&= |2 b \, S|^{-(q-1)(d + 1)/2} \, q^{-q d/(2 b) - r(d)} \, \pi^{-(q-1)d(d-1)/4} \prod_{i=1}^d \frac{\Gamma\{q/(2 b) + (d-i)/2 + 1\}}{\Gamma^q\{1/(2 b) + (d-i)/2 + 1\}}.
\end{aligned}
\]
Note that
\[
\|K_{\nu(b,d), b S}\|_q^2 = \left\{\int_{\mathcal{S}_{++}^d} K_{\nu(b,d), b S}^q(X) \, \rd X\right\}^{2/q} = \{A_{b,q}(S)\}^{2/q}.
\]
Now, by Stirling's approximation, one has, as $x\to \infty$,
\[
\frac{\sqrt{2\pi} e^{-x} x^{x + 1/2}}{\Gamma(x + 1)} = 1 + \OO(x^{-1}).
\]
Therefore,
\begin{align*}
A_{b,q}(S)
&= |2 b \, S|^{-(q-1)(d + 1)/2} \, q^{(q/2 - 1) r(d)} \, \pi^{-(q-1)d(d-1)/4} (2\pi)^{-(q-1)d/2} \\
&\quad\times \prod_{i=1}^d \frac{e^{(q-1)(d-i)/2} \{q/(2 b) + (d-i)/2\}^{q/(2 b) + (d-i+1)/2}}{q^{q/(2 b) + q (d-i+1)/2} \{1/(2 b) + (d-i)/2\}^{q/(2 b) + q (d-i+1)/2}} \, \{1 + \OO_q(b)\} \\
&= |2 b \, S|^{-(q-1)(d + 1)/2} \, q^{(q/2 - 1) r(d)} \, \pi^{-(q-1)d(d-1)/4} (2\pi)^{-(q-1)d/2} \{q/(2 b)\}^{-(q-1) r(d)/2} \\
&\quad\times \prod_{i=1}^d e^{(q-1)(d-i)/2} \left\{1 - \frac{(q-1) (d-i)/2}{q/(2 b) + q (d-i)/2}\right\}^{q/(2 b) + (d-i+1)/2} \, \{1 + \OO_q(b)\} \\
&= |2 b \, S|^{-(q-1)(d + 1)/2} \, q^{(q/2 - 1) r(d)} \, \pi^{-(q-1)d(d-1)/4} (2\pi)^{-(q-1)d/2} \{q/(2 b)\}^{-(q-1) r(d)/2} \{1 + \OO_q(b)\} \\
&= \frac{b^{-(q-1) r(d)/2} |S|^{-(q-1)(d + 1)/2}}{2^{(q-1) r(d)/2 + (q-1) d/2} q^{r(d)/2} \pi^{(q-1)r(d)/2}} \, \{1 + \OO_q(b)\},
\end{align*}
where the last equality follows from the fact that $e^x (1 - x/n)^n = 1 + \OO_x(n^{-1})$ for all $x\in \R$. Taking this last expression to the power $2/q$, and noticing that $(q-1)/q = 1/p$, yields the conclusion.
\end{proof}

The fourth lemma presents exponential bounds on the probabilities that the largest eigenvalue of a Wishart random matrix is excessively large and the smallest eigenvalue is excessively small.

\begin{lemma}\label{lem:eigenvalue.exponential.bounds}\addcontentsline{toc}{subsection}{Lemma~\ref{lem:eigenvalue.exponential.bounds}}
Let $\nu\in (d-1,\infty)$ and $\Sigma\in \mathcal{S}_{++}^d$ be given, and assume that $\mathfrak{X}\sim \mathrm{Wishart}_d(\nu,\Sigma)$. Then, for every scalar $\delta \in (0, 1/\{6 \nu d \, \lambda_1(\Sigma)\})$, one has
\[
\PP\{\lambda_1(\mathfrak{X}) \geq \delta^{-1}\} \leq \exp\{-\delta^{-1} \lambda_1^{-1}(\Sigma)/4\}.
\]
Moreover, for all $\nu \geq d + 1$ and all $\delta\in (0,\infty)$, one has
\[
\PP\{\lambda_d(\mathfrak{X}) \leq \delta\} \leq 1 - \exp\{- \delta \, \tr(\Sigma^{-1})/2\}.
\]
\end{lemma}

\begin{proof}[\bf Proof of Lemma~\ref{lem:eigenvalue.exponential.bounds}]
Let $\mathfrak{Y}\sim \mathrm{Wishart}_d(\nu,I_d)$ and note that $\mathfrak{X} \stackrel{\mathrm{law}}{=} \Sigma^{1/2} \mathfrak{Y} \Sigma^{1/2}$. Let $\bb{v}_1,\ldots,\bb{v}_d$ be an orthonormal eigenbasis of $\Sigma$, so that $\Sigma = \sum_{i=1}^d \lambda_i(\Sigma) \bb{v}_i \bb{v}_i^{\top}$. Then
\[
\tr(\mathfrak{X}) = \tr(\Sigma \, \mathfrak{Y}) = \sum_{i=1}^d \lambda_i(\Sigma) \bb{v}_i^{\top} \mathfrak{Y} \bb{v}_i \stackrel{\mathrm{law}}{=} \sum_{i=1}^d \lambda_i(\Sigma) K_i,
\]
where $K_1, \dots, K_d\stackrel{\mathrm{iid}}{\sim} \chi_{\nu}^2$; see \citet[Theorem~3.2.5]{MR652932} for a justification of the last equality. Given that $\sum_{i=1}^d K_i\sim \chi_{\nu d}^2$, one deduces that
\[
\PP\{\lambda_1(\mathfrak{X}) \geq \delta^{-1}\} \leq \PP\{\tr(\mathfrak{X}) \geq \delta^{-1}\} \leq \PP\{\chi_{\nu d}^2 \geq \delta^{-1} \lambda_1^{-1}(\Sigma)\}.
\]
By applying a Chernoff bound, Theorem~1 of \citet{MR4228660} showed that, for any $a > p > 0$,
\[
\PP(\chi_p^2 > a) \leq \exp\left[-\frac{p}{2} \left\{\frac{a}{p} - 1 - \log\left(\frac{a}{p}\right)\right\}\right].
\]
Since $x - 1 - \log(x) \geq x/2$ for all $x\in (6,\infty)$, it follows that, for $\delta^{-1} \lambda_1^{-1}(\Sigma) > 6 \nu d$,
\[
\PP\{\chi_{\nu d}^2 \geq \delta^{-1} \lambda_1^{-1}(\Sigma)\} \leq \exp\{-\delta^{-1} \lambda_1^{-1}(\Sigma)/4\},
\]
which proves the first claim of the lemma.

To prove the second claim of the lemma, note that
\[
\PP\{\lambda_d(\mathfrak{X}) \leq \delta\} = 1 - \PP\{\lambda_d(\mathfrak{X}) > \delta\} = 1 - \PP(\mathfrak{X} - \delta I_d\in \mathcal{S}_{++}^d).
\]
Denote by $C_d(\nu,\Sigma)$ the normalizing constant of the $\mathrm{Wishart}_d(\nu,\Sigma)$ distribution. Then
\[
\begin{aligned}
\PP(\mathfrak{X} - \delta I_d\in \mathcal{S}_{++}^d)
&= C_d(\nu,\Sigma) \int_{X - \delta I_d\in \mathcal{S}_{++}^d} |X|^{\nu/2 - (d + 1)/2} \exp\{-\tr(\Sigma^{-1}X)/2\} \, \rd X \\
&= C_d(\nu,\Sigma) \int_{\mathcal{S}_{++}^d} |Y + \delta I_d|^{\nu/2 - (d + 1)/2} \exp[-\tr\{\Sigma^{-1} (Y + \delta I_d)\}/2] \, \rd Y \\
&= \exp\{-\delta \, \tr(\Sigma^{-1})/2\} \, C_d(\nu,\Sigma) \int_{\mathcal{S}_{++}^d} |Y + \delta I_d|^{\nu/2 - (d + 1)/2} \exp\{- \tr(\Sigma^{-1}Y)/2\} \, \rd Y.
\end{aligned}
\]
Given that $\nu\geq d + 1$ and $\delta\in (0,\infty)$, one has $|Y + \delta I_d|^{\nu/2 - (d + 1)/2} \geq |Y|^{\nu/2 - (d + 1)/2}$, and thus
\[
\begin{aligned}
\PP(\mathfrak{X} - \delta I_d\in \mathcal{S}_{++}^d)
&\geq \exp\{-\delta \, \tr(\Sigma^{-1})/2\} \, C_d(\nu,\Sigma) \int_{\mathcal{S}_{++}^d} |Y|^{\nu/2 - (d + 1)/2} \exp\{- \tr(\Sigma^{-1}Y)/2\} \, \rd Y \\
&= \exp\{-\delta \, \tr(\Sigma^{-1})/2\}.
\end{aligned}
\]
Therefore
\[
\PP\{\lambda_d(\mathfrak{X}) < \delta\} \leq 1 - \exp\{- \delta \, \tr(\Sigma^{-1})/2\}.
\]
This concludes the proof.
\end{proof}

The fifth lemma contains explicit expressions for the integrals of the squared Wishart KDE over $\mathcal{S}_{++}^d$ and the squared Gaussian kernel estimator over $\mathcal{S}^d$ for matrix log-observations. These expressions are utilized in Appendix~\ref{app:bandwidth.selection} to efficiently select bandwidths for the simulations.

\begin{lemma}\label{lem:integrals.squared.estimators}\addcontentsline{toc}{subsection}{Lemma~\ref{lem:integrals.squared.estimators}}
Let $d\in \N$, $b\in (0,1)$, $S\in \mathcal{S}_{++}^d$, and $Y\in \mathcal{S}^d$ be given. Assume that the observations $\mathfrak{X}_1,\ldots,\mathfrak{X}_n$ have a common marginal density $f$. Recall that
\[
\hat{f}_{n,b}^{\,\mathrm{W}}(S) = \frac{1}{n} \sum_{t=1}^n K_{1/b + d + 1, b S}(\mathfrak{X}_t), \quad
\hat{g}_{n,b}(Y) = \frac{1}{n} \sum_{t=1}^n G_{\log(\mathfrak{X}_t),b}(Y),
\]
where the definitions of $K$ and $G$ are stated at the beginning of Appendix~\ref{app:simulations}. Then
\[
\begin{aligned}
I_1
&\equiv \int_{\mathcal{S}_{++}^d} \{\hat{f}_{n,b}^{\,\mathrm{W}}(S)\}^2 \, \rd S
= \frac{1}{n^2} \sum_{s,t=1}^n \frac{\Gamma_d\{1/b + (d + 1)/2\}}{(2 b)^{r(d)} [\Gamma_d\{1/(2 b) + (d + 1)/2\}]^2} \frac{|\mathfrak{X}_s \mathfrak{X}_t|^{1/(2 b)}}{\left|\mathfrak{X}_s + \mathfrak{X}_t\right|^{1/b + (d + 1)/2}}, \\
I_2
&\equiv \int_{\mathcal{S}^d} \{\hat{g}_{n,b}(Y)\}^2 \, \rd Y
= \frac{1}{n^2} \sum_{s,t=1}^n \frac{\etr\big([-\{\log(\mathfrak{X}_s)\}^2 - \{\log(\mathfrak{X}_t)\}^2 + \{\log(\mathfrak{X}_s) + \log(\mathfrak{X}_t)\}^2/2] / (2 b)\big)}{(2\pi b)^{r(d)/2} \, 2^{d/2}}.
\end{aligned}
\]
\end{lemma}

\begin{proof}[\bf Proof of Lemma~\ref{lem:integrals.squared.estimators}]
For the Wishart KDE, expand the square of the estimator, and integrate term by term using the normalization constant of the inverse Wishart distribution. One obtains
\begin{align*}
I_1
&= \frac{1}{n^2} \sum_{s,t=1}^n \int_{\mathcal{S}_{++}^d} K_{1/b + d + 1, b S}(\mathfrak{X}_s) K_{1/b + d + 1, b S}(\mathfrak{X}_t) \, \rd S \\
&= \frac{1}{n^2} \sum_{s,t=1}^n \!\frac{|\mathfrak{X}_s|^{1/(2 b)} (2 b)^{-d/(2 b) - r(d)}}{\Gamma_d\{1/(2 b) + (d + 1)/2\}} \frac{|\mathfrak{X}_t|^{1/(2 b)} (2 b)^{-d/(2 b) - r(d)}}{\Gamma_d\{1/(2 b) + (d + 1)/2\}} \int_{\mathcal{S}_{++}^d} \hspace{-3mm}\frac{\etr\{-S^{-1} (\mathfrak{X}_s + \mathfrak{X}_t) / (2 b)\}}{|S|^{1/b + (d + 1)}} \, \rd S \\
&= \frac{1}{n^2} \sum_{s,t=1}^n \!\frac{|\mathfrak{X}_s|^{1/(2 b)} (2 b)^{-d/(2 b) - r(d)}}{\Gamma_d\{1/(2 b) + (d + 1)/2\}} \frac{|\mathfrak{X}_t|^{1/(2 b)} (2 b)^{-d/(2 b) - r(d)}}{\Gamma_d\{1/(2 b) + (d + 1)/2\}} \times \frac{\Gamma_d\{1/b + (d + 1)/2\}}{\left|(\mathfrak{X}_s + \mathfrak{X}_t)/ (2 b)\right|^{1/b + (d + 1)/2}} \\
&= \frac{1}{n^2} \sum_{s,t=1}^n \frac{\Gamma_d\{1/b + (d + 1)/2\}}{(2 b)^{r(d)} [\Gamma_d\{1/(2 b) + (d + 1)/2\}]^2} \frac{|\mathfrak{X}_s \mathfrak{X}_t|^{1/(2 b)}}{\left|\mathfrak{X}_s + \mathfrak{X}_t\right|^{1/b + (d + 1)/2}}.
\end{align*}

Similarly, for the Gaussian kernel estimator, expand the square of the estimator, and integrate term by term using the normalization constant of the symmetric matrix-variate normal distribution. One obtains
\[
\begin{aligned}
I_2
&= \frac{1}{n^2} \sum_{s,t=1}^n \int_{\mathcal{S}^d} G_{\log(\mathfrak{X}_s),b}(Y) G_{\log(\mathfrak{X}_t),b}(Y) \, \rd Y \\
&= \frac{1}{n^2} \sum_{s,t=1}^n \left\{\frac{1}{(2\pi b)^{r(d)/2} \, 2^{-d(d - 1)/4}}\right\}^2 \int_{\mathcal{S}^d} \etr\left[-\frac{\{Y - \log(\mathfrak{X}_s)\}^2}{2 b}\right] \etr\left[-\frac{\{Y - \log(\mathfrak{X}_t)\}^2}{2 b}\right] \, \rd Y \\
&= \frac{1}{n^2} \sum_{s,t=1}^n \etr\left[-\frac{\{\log(\mathfrak{X}_s)\}^2}{2 b}\right] \etr\left[-\frac{\{\log(\mathfrak{X}_t)\}^2}{2 b}\right] \frac{\etr[\{\log(\mathfrak{X}_s) + \log(\mathfrak{X}_t)\}^2 / (4 b)]}{(2\pi b)^{r(d)/2} \, 2^{-d(d - 1)/4} \, 2^{r(d)/2}} \\
&\hspace{20mm}\times \int_{\mathcal{S}^d} \frac{\etr\big(-[Y - \{\log(\mathfrak{X}_s) + \log(\mathfrak{X}_t)\}/2]^2 / b\big)}{(2\pi b/2)^{r(d)/2} \, 2^{-d(d - 1)/4}} \, \rd Y \\
&= \frac{1}{n^2} \sum_{s,t=1}^n \frac{\etr\big([-\{\log(\mathfrak{X}_s)\}^2 - \{\log(\mathfrak{X}_t)\}^2 + \{\log(\mathfrak{X}_s) + \log(\mathfrak{X}_t)\}^2/2] / (2 b)\big)}{(2\pi b)^{r(d)/2} \, 2^{d/2}} \times 1.
\end{aligned}
\]
This concludes the proof.
\end{proof}

\section{Future directions}\label{sec:future.directions}

Several open questions remain for advancing the Wishart KDE in both theory and practice; the most salient are outlined below.

First, extending the two-sample test of \citet{arXiv:2603.13935v1} to change-point analysis for strongly mixing $\mathcal{S}_{++}^d$-valued time series would be a natural continuation of the present work. Given observations $\mathfrak{X}_1,\ldots,\mathfrak{X}_n$, one could compute Wishart KDEs on the two subsamples separated by a candidate split point and use the resulting $L^2$ discrepancy as a scan statistic over admissible split points. The main theoretical challenges would be to derive the null distribution of the supremum-type statistic under temporal dependence, to establish consistency and localization rates under fixed or local alternatives, and to handle data-driven bandwidths and multiple change points. Such a procedure would provide a fully nonparametric tool for detecting distributional changes in matrix-valued volatility dynamics.

Second, extending the present treatment of temporal dependence to spatial dependence (e.g., random fields or lattice data) would be a natural next step, with applications to settings where $\mathcal{S}_{++}^d$-valued observations vary across space, such as diffusion tensor tractography \citep{doi:10.1016/C2018-0-02520-7} and kriging prediction in geostatistics \citep{MR3459942}.

Third, the application of the Wishart KDE to DTI segmentation and the use of the corresponding Wishart-kernel Nadaraya--Watson regression estimator for DTI restoration \citep{Wang2004phd} are promising avenues. Because many DTI workflows rely on log-Euclidean or affine-invariant frameworks, incorporating Wishart kernel methods could offer a computationally attractive and potentially less biased alternative. Analyzing the theoretical properties of the regression estimator is also a compelling direction in its own right.

Fourth, the uniform strong consistency result in Theorem~\ref{thm:unif.conv} currently exhibits a logarithmic gap relative to the optimal rate. A more careful analysis of the proof, combined with a refined bandwidth choice, could sharpen the convergence rate.

Fifth, motivated by the minimax analyses of \citet{MR4544604} for the Dirichlet KDE on the simplex and \citet{arXiv:2602.10103v2} for the gamma KDE on $(0,\infty)$, it would be natural, albeit technically very challenging, to develop a corresponding minimax theory for the Wishart KDE on $\mathcal{S}_{++}^d$ under $L^p$ loss and $\beta$-H\"older smoothness of the target. Establishing matching upper and lower risk bounds would identify the ranges of smoothness $\beta$ and loss index $p$ for which the Wishart KDE achieves the minimax rate, and clarify whether there are regimes where it is suboptimal.

Sixth, existing deconvolution results \citep{inria-00632882v1,MR2838725,MR3012414} treat the iid case under the affine-invariant metric. Extending these techniques to dependent positive definite matrix data, particularly time series and spatial fields, remains an open problem. For time dependence, \citet{MR2430252} offers a potential starting point in the univariate setting by quantifying how short- and long-range dependence influence rates and bandwidth selection.

Seventh, high-dimensional covariance matrices arise routinely in finance. Developing efficient implementations (e.g., parallelization, low-rank approximations) and robust procedures (e.g., for incomplete or noisy observations) would broaden the practical scope of the Wishart KDE.

By tackling these research questions, the Wishart KDE can be refined and extended to a wider range of statistical applications. Although the innovative solution described in this paper already provides a practical, boundary-aware estimator on $\mathcal{S}_{++}^d$ under the Euclidean metric, with favorable finite-sample performance relative to the log-Gaussian KDE and the naive Gaussian KDE, further methodological and theoretical developments are needed to fully exploit its potential in high-dimensional, spatially or temporally correlated settings, including DTI and financial econometrics.

\begin{appendices}

\section{Simulations}\label{app:simulations}

This appendix describes the results of a Monte Carlo study which was conducted to compare the performance of the Wishart KDE, defined in \eqref{eq:Wishart.KDE}, against the Gaussian and log-Gaussian KDEs, respectively defined, for every $S\in \mathcal{S}_{++}^d$, by
\[
\hat{f}_{n,b}^{\,\mathrm{G}}(S) = \frac{1}{n} \sum_{t=1}^n G_{\mathfrak{X}_t,b}(S), \quad
\hat{f}_{n,b}^{\,\mathrm{LG}}(S) = \hat{g}_{n,b}\{\log(S)\} \mathcal{J}(S),
\]
where, for every $b \in (0,\infty)$ and $M, Y\in \mathcal{S}^d$,
\[
\hat{g}_{n,b}(Y) = \frac{1}{n} \sum_{t=1}^n G_{\log(\mathfrak{X}_t),b}(Y), \quad
G_{M\!,b}(Y) = \frac{\etr\{-(Y - M)^2 / (2 b)\}}{(2\pi b)^{d(d + 1)/4} \, 2^{-d(d - 1)/4}},
\]
and where the Jacobian determinant of the matrix-log transformation $S\mapsto \log(S)$ is equal to
\[
\mathcal{J}(S) = \frac{1}{|S|} \prod_{1 \leq i < j \leq d} \left[\left|\frac{\log\{\lambda_i(S)\} - \log\{\lambda_j(S)\}}{\lambda_i(S) - \lambda_j(S)}\right| \ind_{\{\lambda_i(S) \neq \lambda_j(S)\}} + \frac{1}{\lambda_i(S)} \ind_{\{\lambda_i(S) = \lambda_j(S)\}}\right];
\]
see \citet[Proposition~1]{MR3580425}. The estimator $\smash{\hat{f}_{n,b}^{\,\mathrm{LG}}}$ was originally proposed in the more general setting of non-compact Riemannian symmetric spaces; see Section~4.1 of \citet{arXiv:2009.01983v3}.

\begin{remark}
The log-Gaussian estimator $\smash{\hat{f}_{n,b}^{\,\mathrm{LG}}}$ is the most natural competitor to $\smash{\hat{f}_{n,b}^{\,\mathrm{W}}}$ as it avoids spill-over issues near the boundary, unlike the Gaussian estimator $\smash{\hat{f}_{n,b}^{\,\mathrm{G}}}$. In fact, it is the only other kernel estimator with this property that could be identified in the literature. Given the structural complexity of the boundary of $\mathcal{S}_{++}^d$, implementing boundary kernels in this setting is notably challenging and left open for future research.
\end{remark}

\begin{remark}
Bandwidth estimation and kernel smoothing for higher-dimensional matrices, say $d \le 5$, pose no particular challenge and are straightforward. However, Monte Carlo experiments studying the performance of the estimators are computationally prohibitive, as they require high-dimensional numerical integration over eigenvalues and rotations. Even when $d = 3$ (which involves six integration variables), several integration routines were found to yield results that differ by several orders of magnitude depending on the adaptation scheme and fineness of the grid, owing to the curse of dimensionality. Monte Carlo estimates of the integrals are also too noisy to be of much practical use unless the number of Monte Carlo draws is astronomical. As such, the simulations presented herein are restricted to the case $d = 2$.
\end{remark}

\subsection{Models}\label{app:models}

For an integer $\kappa \geq d$, consider the $\mathrm{WAR}(1)$ process introduced by \citet{MR1132135}, and studied by \citet{MR2535514}, which consists of the sequence of random positive definite matrices
\[
\mathfrak{X}_t = \sum_{k=1}^{\kappa} \mathfrak{Z}_{k,t} \mathfrak{Z}_{k,t}^{\top}, \quad t\in \N,
\]
where the processes $(\mathfrak{Z}_{1,t})_{t\in \N},\ldots,(\mathfrak{Z}_{\kappa,t})_{t\in \N}$ are independent stationary Gaussian vector autoregressive processes with the same autoregressive coefficient matrix $M\in \smash{\R^{d\times d}}$ and innovation covariance matrix $\Sigma\in \smash{\mathcal{S}_{++}^d}$, with the spectral radius of $M$ smaller than $1$, i.e., for every $k\in \{1,\ldots,\kappa\}$,
\vspace{-2.5mm}
\[
\mathfrak{Z}_{k,t} = M \mathfrak{Z}_{k,t-1} + \e_{k,t}, \quad \e_{k,1},\e_{k,2},\ldots\stackrel{\mathrm{iid}}{\sim} \mathcal{N}_d(\bb{0}_d, \Sigma).
\]
For the process $(\mathfrak{X}_t)_{t=1}^{\infty}$, the stationary marginal density is known to be the Wishart density $K_{\kappa,\Sigma_{\infty}}$ defined in \eqref{eq:Wishart.kernel}, where $\Sigma_{\infty}$ is the solution of the discrete time Lyapunov equation given by $\Sigma_{\infty} = M \Sigma_{\infty} M^{\top} + \Sigma$.

For the simulation, consider the autoregressive coefficient matrices
\[
M_1 =
\begin{pmatrix}
0.9 & 0 \\
1 & 0
\end{pmatrix}, \quad
M_2 =
\begin{pmatrix}
0.3 & -0.3 \\
-0.3 & 0.3
\end{pmatrix}, \quad
M_3 =
\begin{pmatrix}
0.5 & 0 \\
0 & 0.5
\end{pmatrix},
\]
and the innovation covariance matrices
\[
\Sigma_1 =
\begin{pmatrix}
1 & 0.9 \\
0.9 & 1
\end{pmatrix}, \quad
\Sigma_2 =
\begin{pmatrix}
1 & 0.95 \\
0.95 & 1
\end{pmatrix}, \quad
\Sigma_3 =
\begin{pmatrix}
1 & 0.99 \\
0.99 & 1
\end{pmatrix},
\]
corresponding to stronger and stronger degrees of correlation between the coordinates of the innovations. In particular, note that the matrices $\Sigma_1$, $\Sigma_2$, and $\Sigma_3$ are closer and closer to being singular, hence closer and closer to the boundary of $\mathcal{S}_{++}^d$. The matrices $M_1$, $M_2$, and $M_3$ calibrate the persistence of the process; they are chosen as in \citet[p.~170]{MR2535514}. The target densities in the simulation study are the stationary marginal densities for the nine combinations of parameters $(M_i, \Sigma_j)_{i,j=1}^3$, with $\kappa = 4$ degrees of freedom.

\subsection{Bandwidth selection}\label{app:bandwidth.selection}

The first bandwidth selection procedure investigated here is $h$-lag (or $h$-block) least-squares cross-validation (LSCV). The general idea is to choose the bandwidth $b_n\in (0,\infty)$ that minimizes the non-constant part of (an estimate of) the mean integrated squared error. Specifically, for the Wishart kernel, one has
\[
b_n(\mathrm{W}_{\mathrm{lscv}}) = \mathop{\argmin}_{b\in (0,\infty)} \left[\int_{\mathcal{S}_{++}^d} \{\hat{f}_{n,b}^{\,\mathrm{W}}(S)\}^2 \, \rd S - 2 \, \EE\{\hat{f}_{n,b}^{\,\mathrm{W}}(\mathfrak{X}_1)\}\right].
\]
For the log-Gaussian kernel, an explicit expression for the first integral is only possible by working with the estimate of the mean integrated squared error of the corresponding Gaussian KDE $\hat{g}_{n,b}$ on $\mathcal{S}^d$ using the matrix-log observations, viz.
\[
b_n(\mathrm{LG}_{\mathrm{lscv}}) = \mathop{\argmin}_{b\in (0,\infty)} \left(\int_{\mathcal{S}^d} \{\hat{g}_{n,b}(Y)\}^2 \, \rd Y - 2 \, \EE[\hat{g}_{n,b}\{\log(\mathfrak{X}_1)\}]\right).
\]
For the Gaussian kernel, the LSCV bandwidth is obtained from the analogous Euclidean criterion on $\mathcal{S}^d$, using the original observations $\mathfrak{X}_t$ instead of the matrix-log observations $\log(\mathfrak{X}_t)$. In all cases, the expectation term in the $\mathop{\argmin}$ can be approximated by Monte Carlo using the random sample at hand. The integrals of the squared estimators admit explicit expressions (see Lemma~\ref{lem:integrals.squared.estimators}), yielding the following finite-sample approximations:
\[
\begin{aligned}
b_n(\mathrm{W}_{\mathrm{lscv}})
&\approx \mathop{\argmin}_{b\in (0,\infty)} \Bigg[\frac{1}{n^2} \sum_{s,t=1}^n \frac{\Gamma_d\{1/b + (d + 1)/2\}}{(2 b)^{r(d)} [\Gamma_d\{1/(2 b) + (d + 1)/2\}]^2} \frac{|\mathfrak{X}_s \mathfrak{X}_t|^{1/(2 b)}}{\left|\mathfrak{X}_s + \mathfrak{X}_t\right|^{1/b + (d + 1)/2}}\Bigg. \\[-2mm]
&\quad\hspace{84mm}\Bigg.- \frac{2}{n} \sum_{s=1}^n \frac{1}{n_{h,s}} \sum_{\substack{t=1 \\ |s-t|\geq h}}^n K_{\nu(b,d),b \mathfrak{X}_s}(\mathfrak{X}_t)\Bigg], \\
b_n(\mathrm{LG}_{\mathrm{lscv}})
&\approx \mathop{\argmin}_{b\in (0,\infty)} \Bigg[\frac{1}{n^2} \sum_{s,t=1}^n \frac{\etr\big([-\{\log(\mathfrak{X}_s)\}^2 - \{\log(\mathfrak{X}_t)\}^2 + \{\log(\mathfrak{X}_s) + \log(\mathfrak{X}_t)\}^2/2] / (2 b)\big)}{(2\pi b)^{r(d)/2} \, 2^{d/2}}\Bigg. \\[-3mm]
&\quad\hspace{76.5mm}\Bigg.- \frac{2}{n} \sum_{s=1}^n \frac{1}{n_{h,s}} \sum_{\substack{t=1 \\ |s-t|\geq h}}^n G_{\log(\mathfrak{X}_t),b}\{\log(\mathfrak{X}_s)\}\Bigg], \\
b_n(\mathrm{G}_{\mathrm{lscv}})
&\approx \mathop{\argmin}_{b\in (0,\infty)} \Bigg[\frac{1}{n^2} \sum_{s,t=1}^n \frac{\etr\big[\{-\mathfrak{X}_s^2 - \mathfrak{X}_t^2 + (\mathfrak{X}_s + \mathfrak{X}_t)^2/2\} / (2 b)\big]}{(2\pi b)^{r(d)/2} \, 2^{d/2}} - \frac{2}{n} \sum_{s=1}^n \frac{1}{n_{h,s}} \!\! \sum_{\substack{t=1 \\ |s-t|\geq h}}^n \!\! G_{\mathfrak{X}_t,b}(\mathfrak{X}_s)\Bigg],
\end{aligned}
\]
where, for every $s\in [n]$, $n_{h,s}$ is the number of indices $t\in [n]$ such that $|s - t|\geq h$. In line with \citet{MR2756423}, the results are reported with the lag parameter $h = \lceil n^{1/4} \rceil$.

The second bandwidth selection procedure included in the study maximizes the leave-one-out likelihood cross-validation (LCV) criterion for the kernels $\mathrm{K}\in\{\mathrm{W},\mathrm{LG},\mathrm{G}\}$:
\[
b_n(\mathrm{K}_{\mathrm{lcv}}) = \mathop{\argmax}_{b\in (0,\infty)} \frac{1}{n}\sum_{t=1}^{n} \log\big\{\hat{f}^{\,\mathrm{K}}_{-t,b}(\mathfrak{X}_t)\big\},
\]
where the subscript ``$-t$'' indicates that the $t$th observation is removed to compute the estimate.

\begin{remark}
To evaluate the $h$-lag LSCV and LCV criteria while avoiding numerical overflow, sums of positive terms are computed on a log-scale, as well as all normalizing constants. More specifically, for a generic sum $x_1 + \cdots + x_B$ with $x_k \in (0,\infty)$, let $\ell_{\max} = \max \{ \log(x_1), \dots, \log(x_B) \} $, then one computes
\[
\log\left( \sum_{k=1}^B x_k\right) = \ell_{\max} + \log \left[\sum_{k=1}^B \exp\{\log(x_k) - \ell_{\max}\}\right];
\]
see, e.g., Lemma~5.1 of \citet{MR3089615}.
\end{remark}

\subsection{Results}\label{app:results}

Let $R = 1024$ be the number of replications. For each sample size $n\in \{100,200,300\}$, each method $\mathrm{m}\in \{\mathrm{W}_{\mathrm{lscv}},\mathrm{W}_{\mathrm{lcv}},\mathrm{LG}_{\mathrm{lscv}},\mathrm{LG}_{\mathrm{lcv}},\mathrm{G}_{\mathrm{lscv}},\mathrm{G}_{\mathrm{lcv}}\}$, and each $\mathrm{WAR}(1)$ model $\smash{(M_i,\Sigma_j)_{i,j=1}^3}$ with $\kappa = 4$ degrees of freedom, the root integrated squared error (RISE) is calculated for each random sample $\smash{\mathfrak{X}_1^{(r)},\ldots,\mathfrak{X}_n^{(r)}}$ with $r \in \{1, \dots, R\}$, namely
\[
\mathrm{RISE}_{b_n\hspace{-0.3mm}(\mathrm{m})}^{i,j,\mathrm{m}}(r) = \sqrt{\int_{\mathcal{S}_{++}^d} \big\{\hat{f}_{n,b_n\hspace{-0.3mm}(\mathrm{m})}^{\,\mathrm{m}}(S) - f_{i,j}(S)\big\}^2 \, \rd S},
\]
where $f_{i,j}$ is the marginal density corresponding to the $(M_i,\Sigma_j)$-model.

Let $O(d)$ denote the orthogonal group of order $d$, $ \, \rd H$ the Haar probability measure on $O(d)$, and $\Lambda = \mathrm{diag}(\bb{\lambda})$. Weyl's integration formula (see, e.g., Theorem~3.2.17 of \citet{MR652932}) shows that the above integral is equal to
\[
\frac{\pi^{d^2/2}}{\Gamma_d(d/2) d!} \int_{[0,\infty)^d} \int_{O(d)} \big|\hat{f}_{n,b_n\hspace{-0.3mm}(\mathrm{m})}^{\,\mathrm{m}}(H \Lambda H^{\top}) - f_{i,j}(H \Lambda H^{\top})\big|^2 \, \prod_{1 \leq k < \ell \leq d} |\lambda_k - \lambda_{\ell}| \, \rd H \, \rd \bb{\lambda},
\]
where the components of the vector $\bb{\lambda} = (\lambda_1,\ldots,\lambda_d)$ are assumed to be unordered for simplicity. In the case $d=2$, the expression simplifies to
\[
\frac{1}{4} \int_{0}^{\infty} \int_{0}^{\infty} \int_0^{2\pi} \big|\hat{f}_{n,b_n\hspace{-0.3mm}(\mathrm{m})}^{\,\mathrm{m}}\{R(\theta) \Lambda R(\theta)^{\top}\} - f_{i,j}\{R(\theta) \Lambda R(\theta)^{\top}\}\big|^2 \, |\lambda_1 - \lambda_2| \, \rd \theta \, \rd \lambda_1 \, \rd \lambda_2,
\]
where $\Lambda = \mathrm{diag}(\lambda_1,\lambda_2)$, and $R(\theta)$ denotes the rotation matrix
\[
R(\theta) =
\begin{pmatrix}
\cos\theta & -\sin\theta \\
\sin\theta & \cos\theta
\end{pmatrix}, \quad \theta\in [0,2\pi).
\]

The simulation results, presented in Figure~\ref{fig:RISE.boxplots} and Table~\ref{tab:summary}, are briefly summarized below.
\begin{enumerate}
\item The method $\mathrm{W}_{\mathrm{lscv}}$ gives the most favorable overall performance among the six methods. For every pair $(M_i, \Sigma_j)$ and every $n\in \{100,200,300\}$, it has either the smallest median $\mathrm{RISE}$ or the second smallest one. In the latter cases, the smallest median is attained by $\mathrm{LG}_{\mathrm{lcv}}$, but the difference is very modest, whereas the interquartile range (IQR) of $\mathrm{LG}_{\mathrm{lcv}}$ is substantially larger than that of $\mathrm{W}_{\mathrm{lscv}}$. Thus $\mathrm{W}_{\mathrm{lscv}}$ remains preferable in those cases as well, and offers the best compromise between accuracy and stability. The method $\mathrm{LG}_{\mathrm{lcv}}$ is the closest competitor in terms of median $\mathrm{RISE}$, whereas $\mathrm{LG}_{\mathrm{lscv}}$ is markedly more variable and more prone to large $\mathrm{RISE}$ values.

\item The classical Gaussian KDEs are clearly the least competitive class of methods in terms of median $\mathrm{RISE}$ when compared with the boundary-aware $\mathrm{W}$ and $\mathrm{LG}$ competitors, especially as $\Sigma_j$ gets closer to being singular. This deterioration becomes more pronounced from $\Sigma_1$ to $\Sigma_3$, which is consistent with the spill-over and boundary-bias issues of classical symmetric kernels on constrained supports. Indeed, the stationary marginal density is $K_{\kappa,\Sigma_{\infty}}$, where $\Sigma_{\infty}$ is determined by $M_i$ and $\Sigma_j$ through the Lyapunov equation, so increasingly ill-conditioned innovation covariance matrices lead to target densities that are increasingly concentrated near the boundary of $\mathcal{S}_{++}^d$. Although the Gaussian KDEs sometimes have small IQRs, especially for $\Sigma_3$, these small IQRs are accompanied by large median errors and therefore reflect consistently inaccurate estimates rather than good performance.

\item For each fixed $M_i$ and sample size, the median $\mathrm{RISE}$ values increase as the correlation in $\Sigma_j$ increases from $\Sigma_1$ to $\Sigma_3$. The same pattern is visible for the IQRs of the boundary-aware methods, with $\mathrm{W}_{\mathrm{lscv}}$ consistently being the most stable among them. This is again consistent with the fact that the induced stationary scale matrix $\Sigma_{\infty}$ becomes increasingly ill-conditioned, making the target density more concentrated near the boundary of $\mathcal{S}_{++}^d$. The choice of $M_i$ also affects the absolute level of the errors through $\Sigma_{\infty}$, but the relative ranking of the methods is broadly stable across the three autoregressive matrices.

\item As the sample size $n$ increases, the medians and IQRs generally decrease for the boundary-aware methods, as expected. This decrease is clearest for $\mathrm{W}_{\mathrm{lscv}}$, $\mathrm{W}_{\mathrm{lcv}}$, and $\mathrm{LG}_{\mathrm{lcv}}$. For the Gaussian KDEs, the effect of increasing $n$ is less systematic near the boundary: the medians can remain nearly unchanged or even increase slightly, and the IQRs can be misleadingly small when the estimator remains affected by boundary bias. This is particularly visible for the case $(M_1,\Sigma_3)$, where the Gaussian KDEs have large median $\mathrm{RISE}$ values despite very small IQRs.

\item As shown in Table~\ref{tab:mean.time}, mean bandwidth computation times (based on $100$ replications) are relatively stable across the nine models for a fixed sample size, with $n$ being the main driver of computing cost. The fastest method depends on the model and sample size, while $\mathrm{W}_{\mathrm{lscv}}$ is usually among the slowest methods. Even so, after averaging over the nine models, the slowest mean bandwidth computation time is less than about three times the fastest one for each sample size, and the largest mean times are close to one second when $n = 300$, so bandwidth selection is not a computational bottleneck in practice.
\end{enumerate}

\begin{table}[H]
\caption{Mean bandwidth computation time (in milliseconds) for each method, model and sample size. Bold indicates the smallest mean among the six competitors for the corresponding model and sample size.}
\label{tab:mean.time}
\footnotesize
\renewcommand{\arraystretch}{0.8}
\setlength{\tabcolsep}{1.5pt}
\centering
\begin{tabular}{c|cccccc|cccccc|cccccc}
\toprule
 & \multicolumn{6}{c|}{$n = 100$} & \multicolumn{6}{c|}{$n = 200$} & \multicolumn{6}{c}{$n = 300$} \\
\midrule
Model & $W_{\text{\hspace{-0.4mm}lscv}}$ & $W_{\text{\hspace{-0.4mm}lcv}}$ & $\mathrm{LG}_{\text{lscv}}$ & $\mathrm{LG}_{\text{lcv}}$ & $\mathrm{G}_{\text{lscv}}$ & $\mathrm{G}_{\text{lcv}}$ & $W_{\text{\hspace{-0.4mm}lscv}}$ & $W_{\text{\hspace{-0.4mm}lcv}}$ & $\mathrm{LG}_{\text{lscv}}$ & $\mathrm{LG}_{\text{lcv}}$ & $\mathrm{G}_{\text{lscv}}$ & $\mathrm{G}_{\text{lcv}}$ & $W_{\text{\hspace{-0.4mm}lscv}}$ & $W_{\text{\hspace{-0.4mm}lcv}}$ & $\mathrm{LG}_{\text{lscv}}$ & $\mathrm{LG}_{\text{lcv}}$ & $\mathrm{G}_{\text{lscv}}$ & $\mathrm{G}_{\text{lcv}}$ \\
\midrule
$(M_{1}, \Sigma_{1})$ & 110 & \textbf{46} & 73 & 48 & 102 & 120 & 444 & 189 & 295 & \textbf{187} & 190 & 504 & 1029 & 435 & 680 & \textbf{434} & 440 & 1064 \\
$(M_{1}, \Sigma_{2})$ & 109 & \textbf{46} & 74 & 48 & 49 & 107 & 455 & 190 & 296 & \textbf{189} & 191 & 332 & 996 & 423 & 673 & \textbf{422} & 440 & 728 \\
$(M_{1}, \Sigma_{3})$ & 108 & \textbf{46} & 79 & 51 & 49 & 65 & 443 & 185 & 314 & 198 & \textbf{183} & 202 & 1011 & 431 & 714 & 445 & \textbf{402} & 433 \\
$(M_{2}, \Sigma_{1})$ & 114 & 47 & 75 & 48 & 57 & \textbf{31} & 447 & 188 & 288 & 187 & 233 & \textbf{133} & 1017 & 429 & 671 & 426 & 547 & \textbf{311} \\
$(M_{2}, \Sigma_{2})$ & 108 & 45 & 73 & 47 & 58 & \textbf{32} & 442 & 187 & 292 & 188 & 242 & \textbf{136} & 1026 & 427 & 671 & 423 & 568 & \textbf{319} \\
$(M_{2}, \Sigma_{3})$ & 108 & 45 & 73 & 48 & 68 & \textbf{38} & 441 & 186 & 299 & 187 & 281 & \textbf{164} & 1008 & 427 & 670 & 436 & 647 & \textbf{377} \\
$(M_{3}, \Sigma_{1})$ & 109 & 45 & 71 & 47 & 55 & \textbf{32} & 445 & 187 & 292 & 190 & 233 & \textbf{132} & 1024 & 431 & 673 & 420 & 533 & \textbf{303} \\
$(M_{3}, \Sigma_{2})$ & 109 & 45 & 73 & 48 & 58 & \textbf{33} & 459 & 191 & 302 & 195 & 249 & \textbf{142} & 1025 & 433 & 675 & 430 & 564 & \textbf{315} \\
$(M_{3}, \Sigma_{3})$ & 107 & 46 & 74 & 48 & 66 & \textbf{36} & 464 & 196 & 309 & 196 & 282 & \textbf{162} & 1021 & 429 & 680 & 431 & 620 & \textbf{368} \\
\bottomrule
\end{tabular}
\end{table}

\begin{figure}[H]
\centering
\includegraphics[trim={0cm 1cm 0cm 0.4cm}, clip, width=0.96\textwidth]{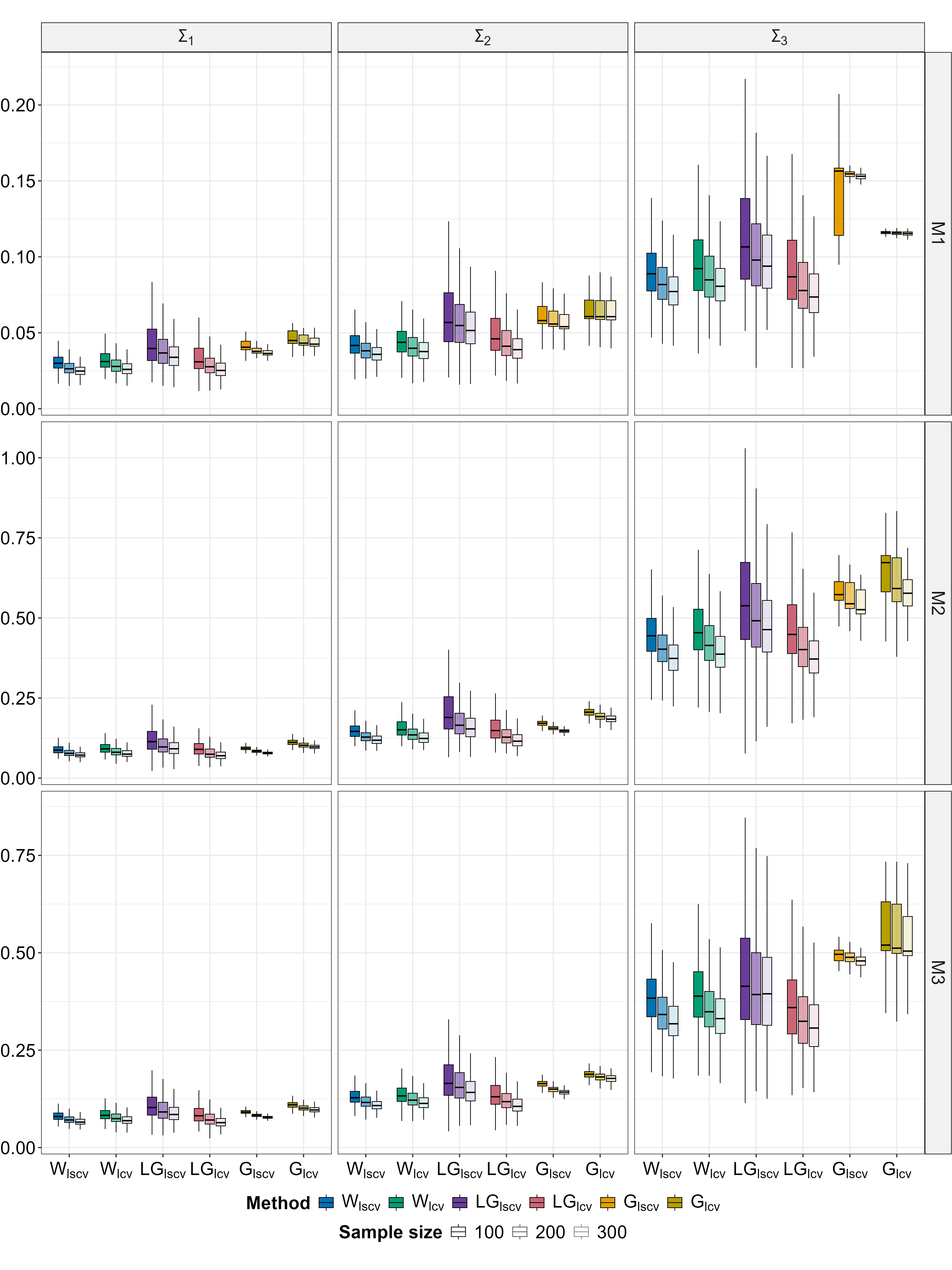}
\caption{Boxplots of the $1024$ $\smash{\mathrm{RISE}}$ values for different configurations of autoregressive coefficient matrix ($M_i, \, i\in \{1,2,3\}$), covariance matrix ($\Sigma_j, \, j\in \{1,2,3\}$), method $\mathrm{m}\in \{\mathrm{W}_{\mathrm{lscv}},\mathrm{W}_{\mathrm{lcv}},\mathrm{LG}_{\mathrm{lscv}},\mathrm{LG}_{\mathrm{lcv}},\mathrm{G}_{\mathrm{lscv}},\mathrm{G}_{\mathrm{lcv}}\}$, and sample sizes $n\in \{100,200,300\}$.}
\label{fig:RISE.boxplots}
\end{figure}

\begin{table}[!ht]
\caption{Comparison of the Wishart (W), log-Gaussian (LG), and Gaussian (G) methods in terms of the median and interquartile range (IQR) of $1024$ RISE values ($\times 10^5$) for different configurations of autoregressive coefficient matrices ($M_i, \, i\in \{1,2,3\}$), innovation covariance matrices ($\Sigma_j, \, j\in \{1,2,3\}$), and sample sizes $n \in \{100, 200, 300\}$. Dark gray cells indicate the smallest value among the six competitors for the corresponding model, sample size, and statistic; light gray cells indicate the second smallest.}
\label{tab:summary}

\footnotesize\renewcommand{\arraystretch}{0.67}
\setlength{\tabcolsep}{8pt}
\centering
{
\begin{tabular}{cc|ccc|ccc}
\toprule
\multicolumn{2}{c|}{} & $n = 100$ & $n = 200$ & $n = 300$ & $n = 100$ & $n = 200$ & $n = 300$ \\
\cmidrule(r){1-2}\cmidrule(lr){3-5}\cmidrule(l){6-8}
Model & Method & \multicolumn{3}{c|}{Median RISE} & \multicolumn{3}{c}{IQR RISE} \\
\midrule
$(M_1,\Sigma_1)$ & $\mathrm{W}_{\mathrm{lscv}}$ & \cellcolor{black!30}2995 & \cellcolor{black!30}2630 & \cellcolor{black!30}2479 & \cellcolor{black!10}719 & \cellcolor{black!10}621 & \cellcolor{black!10}482 \\ 
 & $\mathrm{W}_{\mathrm{lcv}}$ & 3104 & 2785 & 2590 & 889 & 745 & 650 \\ 
 & $\mathrm{LG}_{\mathrm{lscv}}$ & 3964 & 3670 & 3386 & 2069 & 1586 & 1230 \\ 
 & $\mathrm{LG}_{\mathrm{lcv}}$ & \cellcolor{black!10}3087 & \cellcolor{black!10}2763 & \cellcolor{black!10}2520 & 1343 & 960 & 822 \\ 
 & $\mathrm{G}_{\mathrm{lscv}}$ & 4050 & 3773 & 3641 & \cellcolor{black!30}559 & \cellcolor{black!30}325 & \cellcolor{black!30}285 \\ 
 & $\mathrm{G}_{\mathrm{lcv}}$ & 4490 & 4318 & 4253 & 831 & 675 & 534 \\ 
\midrule
$(M_1,\Sigma_2)$ & $\mathrm{W}_{\mathrm{lscv}}$ & \cellcolor{black!30}4167 & \cellcolor{black!30}3813 & \cellcolor{black!30}3580 & \cellcolor{black!10}1154 & \cellcolor{black!30}966 & \cellcolor{black!30}828 \\ 
 & $\mathrm{W}_{\mathrm{lcv}}$ & \cellcolor{black!10}4371 & \cellcolor{black!10}3979 & \cellcolor{black!10}3764 & 1355 & 1225 & 1058 \\ 
 & $\mathrm{LG}_{\mathrm{lscv}}$ & 5681 & 5477 & 5150 & 3222 & 2500 & 2086 \\ 
 & $\mathrm{LG}_{\mathrm{lcv}}$ & 4602 & 4116 & 3890 & 2105 & 1652 & 1282 \\ 
 & $\mathrm{G}_{\mathrm{lscv}}$ & 5805 & 5578 & 5402 & \cellcolor{black!30}1133 & \cellcolor{black!10}1009 & \cellcolor{black!10}940 \\ 
 & $\mathrm{G}_{\mathrm{lcv}}$ & 6069 & 6061 & 6064 & 1225 & 1248 & 1262 \\ 
\midrule
$(M_1,\Sigma_3)$ & $\mathrm{W}_{\mathrm{lscv}}$ & \cellcolor{black!10}8886 & \cellcolor{black!10}8176 & \cellcolor{black!10}7715 & \cellcolor{black!10}2487 & 2104 & 1847 \\ 
 & $\mathrm{W}_{\mathrm{lcv}}$ & 9225 & 8484 & 8062 & 3335 & 2696 & 2136 \\ 
 & $\mathrm{LG}_{\mathrm{lscv}}$ & 10655 & 9790 & 9386 & 5310 & 4083 & 3502 \\ 
 & $\mathrm{LG}_{\mathrm{lcv}}$ & \cellcolor{black!30}8685 & \cellcolor{black!30}7782 & \cellcolor{black!30}7356 & 3892 & 3024 & 2549 \\ 
 & $\mathrm{G}_{\mathrm{lscv}}$ & 15648 & 15463 & 15305 & 4428 & \cellcolor{black!10}313 & \cellcolor{black!10}295 \\ 
 & $\mathrm{G}_{\mathrm{lcv}}$ & 11608 & 11569 & 11540 & \cellcolor{black!30}140 & \cellcolor{black!30}181 & \cellcolor{black!30}208 \\ 
\midrule
$(M_2,\Sigma_1)$ & $\mathrm{W}_{\mathrm{lscv}}$ & \cellcolor{black!30}8745 & \cellcolor{black!10}7746 & \cellcolor{black!10}7180 & 1885 & 1626 & 1281 \\ 
 & $\mathrm{W}_{\mathrm{lcv}}$ & 9137 & 8078 & 7497 & 2383 & 2027 & 1789 \\ 
 & $\mathrm{LG}_{\mathrm{lscv}}$ & 11379 & 9743 & 9157 & 5563 & 4080 & 3401 \\ 
 & $\mathrm{LG}_{\mathrm{lcv}}$ & \cellcolor{black!10}8992 & \cellcolor{black!30}7473 & \cellcolor{black!30}6984 & 3266 & 2572 & 2029 \\ 
 & $\mathrm{G}_{\mathrm{lscv}}$ & 9313 & 8422 & 7845 & \cellcolor{black!30}842 & \cellcolor{black!30}699 & \cellcolor{black!30}562 \\ 
 & $\mathrm{G}_{\mathrm{lcv}}$ & 11166 & 10267 & 9750 & \cellcolor{black!10}1327 & \cellcolor{black!10}1223 & \cellcolor{black!10}1011 \\ 
\midrule
$(M_2,\Sigma_2)$ & $\mathrm{W}_{\mathrm{lscv}}$ & \cellcolor{black!30}14596 & \cellcolor{black!10}12817 & \cellcolor{black!10}11835 & 3268 & 2508 & 2317 \\ 
 & $\mathrm{W}_{\mathrm{lcv}}$ & 15083 & 13504 & 12382 & 4153 & 3243 & 2953 \\ 
 & $\mathrm{LG}_{\mathrm{lscv}}$ & 18952 & 16513 & 15364 & 10064 & 6417 & 5759 \\ 
 & $\mathrm{LG}_{\mathrm{lcv}}$ & \cellcolor{black!10}14854 & \cellcolor{black!30}12790 & \cellcolor{black!30}11536 & 5586 & 4153 & 3483 \\ 
 & $\mathrm{G}_{\mathrm{lscv}}$ & 17141 & 15597 & 14727 & \cellcolor{black!30}1229 & \cellcolor{black!30}995 & \cellcolor{black!30}797 \\ 
 & $\mathrm{G}_{\mathrm{lcv}}$ & 20592 & 19213 & 18408 & \cellcolor{black!10}1814 & \cellcolor{black!10}1869 & \cellcolor{black!10}1781 \\ 
\midrule
$(M_2,\Sigma_3)$ & $\mathrm{W}_{\mathrm{lscv}}$ & \cellcolor{black!30}44424 & \cellcolor{black!10}40248 & \cellcolor{black!10}37370 & \cellcolor{black!10}10264 & \cellcolor{black!10}8305 & \cellcolor{black!10}7923 \\ 
 & $\mathrm{W}_{\mathrm{lcv}}$ & 45386 & 41421 & 38707 & 12667 & 10888 & 9632 \\ 
 & $\mathrm{LG}_{\mathrm{lscv}}$ & 53801 & 49119 & 46359 & 24064 & 19871 & 16117 \\ 
 & $\mathrm{LG}_{\mathrm{lcv}}$ & \cellcolor{black!10}44852 & \cellcolor{black!30}40124 & \cellcolor{black!30}37179 & 15237 & 12274 & 10034 \\ 
 & $\mathrm{G}_{\mathrm{lscv}}$ & 57287 & 54429 & 52591 & \cellcolor{black!30}5839 & \cellcolor{black!30}8119 & \cellcolor{black!30}7500 \\ 
 & $\mathrm{G}_{\mathrm{lcv}}$ & 67270 & 59194 & 57714 & 11334 & 13722 & 8212 \\ 
\midrule
$(M_3,\Sigma_1)$ & $\mathrm{W}_{\mathrm{lscv}}$ & \cellcolor{black!30}7955 & \cellcolor{black!10}7079 & \cellcolor{black!10}6533 & 1639 & 1405 & 1269 \\ 
 & $\mathrm{W}_{\mathrm{lcv}}$ & 8270 & 7420 & 6890 & 2102 & 1935 & 1677 \\ 
 & $\mathrm{LG}_{\mathrm{lscv}}$ & 10277 & 9138 & 8481 & 4614 & 4086 & 3160 \\ 
 & $\mathrm{LG}_{\mathrm{lcv}}$ & \cellcolor{black!10}8163 & \cellcolor{black!30}7072 & \cellcolor{black!30}6391 & 3190 & 2560 & 1903 \\ 
 & $\mathrm{G}_{\mathrm{lscv}}$ & 9118 & 8255 & 7748 & \cellcolor{black!30}688 & \cellcolor{black!30}592 & \cellcolor{black!30}475 \\ 
 & $\mathrm{G}_{\mathrm{lcv}}$ & 10972 & 10112 & 9651 & \cellcolor{black!10}1177 & \cellcolor{black!10}1065 & \cellcolor{black!10}1081 \\ 
\midrule
$(M_3,\Sigma_2)$ & $\mathrm{W}_{\mathrm{lscv}}$ & \cellcolor{black!30}12789 & \cellcolor{black!30}11588 & \cellcolor{black!10}10776 & 2769 & 2379 & 1876 \\ 
 & $\mathrm{W}_{\mathrm{lcv}}$ & 13255 & 12168 & 11325 & 3430 & 3018 & 2533 \\ 
 & $\mathrm{LG}_{\mathrm{lscv}}$ & 16473 & 15435 & 14156 & 7839 & 6504 & 5006 \\ 
 & $\mathrm{LG}_{\mathrm{lcv}}$ & \cellcolor{black!10}13032 & \cellcolor{black!10}11782 & \cellcolor{black!30}10568 & 4841 & 3632 & 3059 \\ 
 & $\mathrm{G}_{\mathrm{lscv}}$ & 16418 & 15002 & 14233 & \cellcolor{black!30}1191 & \cellcolor{black!30}1091 & \cellcolor{black!30}918 \\ 
 & $\mathrm{G}_{\mathrm{lcv}}$ & 18806 & 18122 & 17746 & \cellcolor{black!10}1406 & \cellcolor{black!10}1507 & \cellcolor{black!10}1464 \\ 
\midrule
$(M_3,\Sigma_3)$ & $\mathrm{W}_{\mathrm{lscv}}$ & \cellcolor{black!10}38348 & \cellcolor{black!10}34153 & \cellcolor{black!10}31769 & \cellcolor{black!10}9646 & \cellcolor{black!10}8164 & \cellcolor{black!10}7548 \\ 
 & $\mathrm{W}_{\mathrm{lcv}}$ & 38867 & 34849 & 33095 & 11652 & 9076 & 8911 \\ 
 & $\mathrm{LG}_{\mathrm{lscv}}$ & 41400 & 39278 & 39464 & 20892 & 18455 & 17452 \\ 
 & $\mathrm{LG}_{\mathrm{lcv}}$ & \cellcolor{black!30}35924 & \cellcolor{black!30}32412 & \cellcolor{black!30}30671 & 13841 & 11980 & 10726 \\ 
 & $\mathrm{G}_{\mathrm{lscv}}$ & 49593 & 48822 & 47898 & \cellcolor{black!30}2736 & \cellcolor{black!30}2240 & \cellcolor{black!30}2111 \\ 
 & $\mathrm{G}_{\mathrm{lcv}}$ & 51964 & 51181 & 50429 & 12542 & 12665 & 10038 \\ 
\bottomrule
\end{tabular}}
\end{table}

\clearpage

\section{Real-data application}\label{app:application}

This section illustrates the use of the Wishart KDE for estimating the marginal density of realized covariance matrices from financial time series data, constructed from intraday returns.

More specifically, using a Bloomberg terminal, 5-minute intraday price data (the last trading price for each of the 192 5-minute intervals from 4 AM to 8 PM ET) on trading days were retrieved for the period from September 13, 2023, to September 12, 2024 (one calendar year) for both the AMZN stock and the SPY exchange-traded fund. For each trading day, the $2 \times 2$ realized covariance matrix of the return levels was computed from the 5-minute intraday prices of AMZN and SPY, where the component in position $(1,1)$ denotes the realized variance of AMZN and the component in position $(2,2)$ denotes the realized variance of SPY. This resulted in a time series of $n = 250$ covariance matrices, whose time evolution is shown in Figure~\ref{fig:s11.s12.s22.processes} through the realized variances and correlation series. There are visible spikes in volatility around April and August 2024.

\begin{figure}[b!]
\centering
\includegraphics[trim={1cm 0.3cm 0.4cm 1cm}, clip, width=0.96\textwidth]{application1/fig/rvar_vcov_cor.pdf}
\vspace{-5mm}
\caption{Time series of the realized variances of AMZN and SPY return levels (left) and of the realized correlation between AMZN and SPY (right).}
\label{fig:s11.s12.s22.processes}
\end{figure}

To estimate the marginal density of this time series, the Wishart KDE, $\smash{\hat{f}_{n,b}^{\,\mathrm{W}}}$, was used. The bandwidth $b$ was selected using the $h$-lag least-squares cross-validation method described in Appendix~\ref{app:bandwidth.selection}, with $h = \lceil n^{1/4} \rceil=4$. The left panel of Figure~\ref{fig:app1.CV} shows the negative logarithm of the least-squares cross-validation criterion. The optimal bandwidth under this criterion is $b_n^{\star} \approx 0.021$.

Contour plots of $S\mapsto \smash{\log\{\hat{f}_{n,b_n^{\star}}^{\,\mathrm{W}}(S)\}}$ are displayed in Figure~\ref{fig:app1.contour} for fixed correlation levels of $0$, $0.25$, $0.5$ and $0.75$; the figures are rescaled so that the logarithm of the density equals zero at the maximum value. High positive correlation is compatible with a wider range of combinations of intraday variances. The right panel of Figure~\ref{fig:app1.CV} shows the estimated correlation as a function of the realized variance quantile levels of AMZN and SPY. According to the kernel density estimate, the correlation is generally higher in periods of high volatility, but this relationship is asymmetric, as seen from the shape of the contour plots in the right panel of Figure~\ref{fig:app1.contour}. Higher realized variances of the SPY index return levels translate into high correlations with the Amazon stock, but not the converse. This aligns with AMZN's composition of over 4\% of the S\&P 500 index (in 2023--2024) and its well-established correlation with other technology companies (Microsoft, Apple, etc.), which constitute a large portion of the index.

\begin{figure}[ht!]
\centering
\includegraphics[trim={0.4cm 0.4cm 0.2cm 1cm}, clip, width=0.98\textwidth]{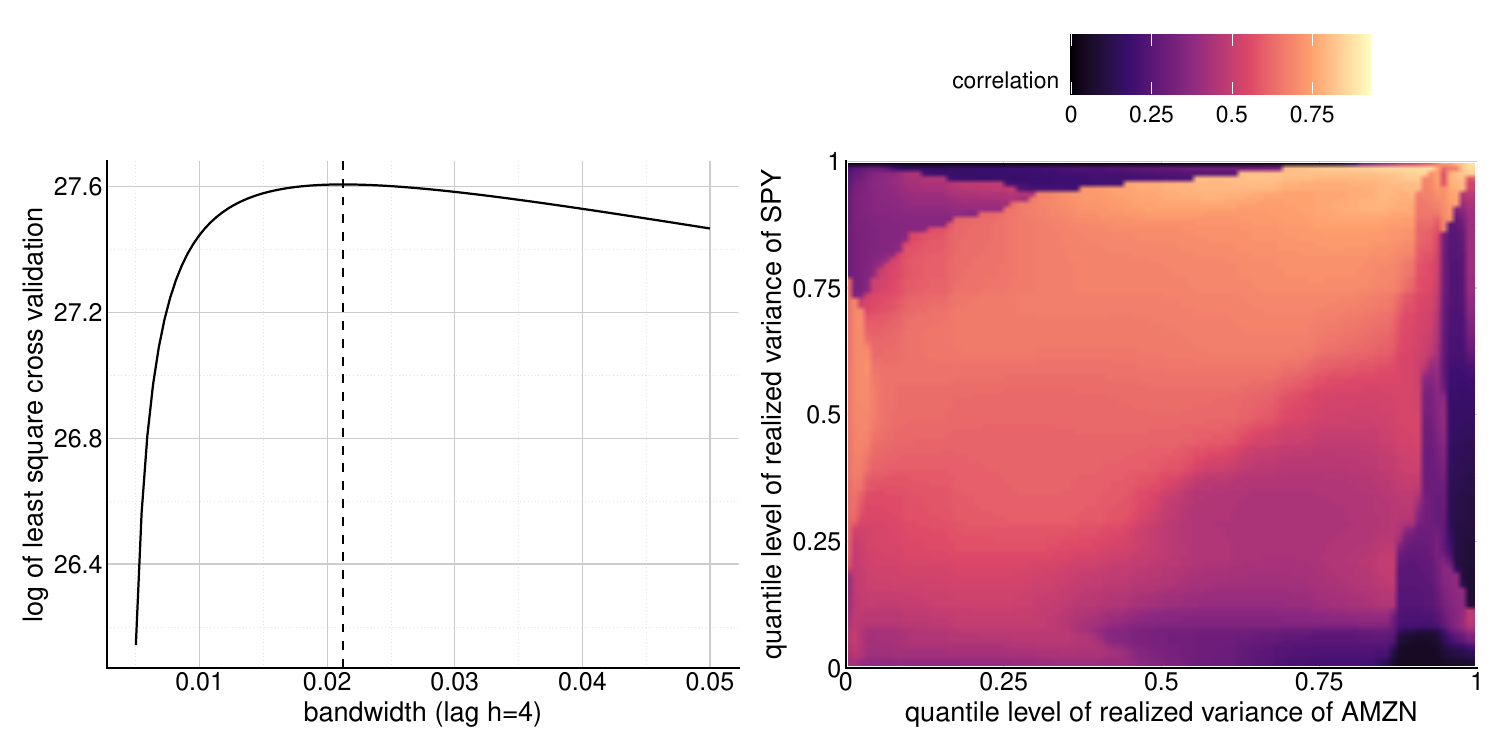}
\caption{Negative logarithm of the least-squares cross-validation criterion as a function of the bandwidth $b$ (left), and estimated correlation as a function of the realized variance quantile levels (right).}
\label{fig:app1.CV}
\end{figure}

\begin{figure}[H]
\vspace{-4mm}
\centering
\includegraphics[trim={0.3cm 0.3cm 0.5cm 0cm}, clip, width=0.495\textwidth]{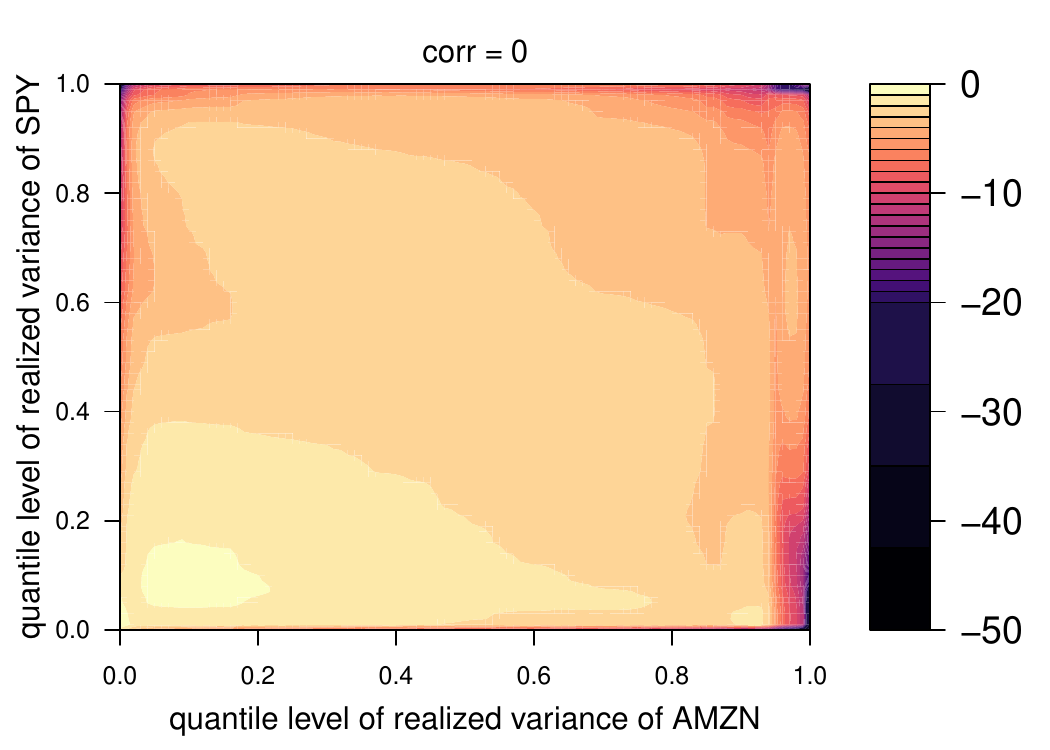}\hfill
\includegraphics[trim={0.3cm 0.3cm 0.5cm 0cm}, clip, width=0.495\textwidth]{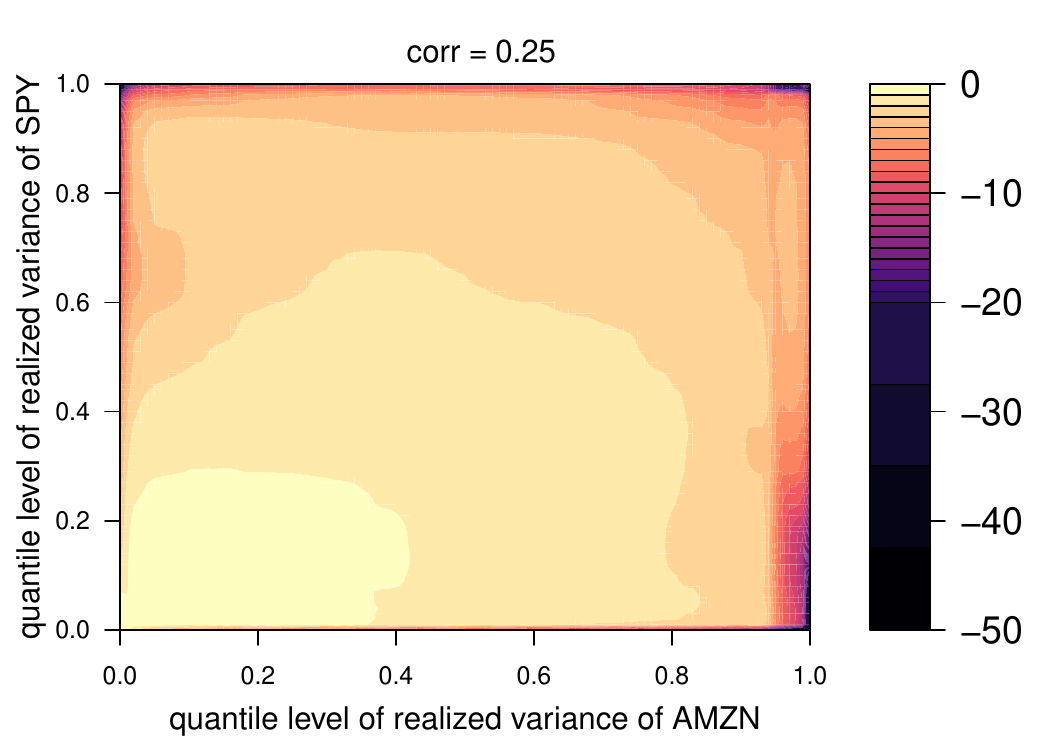}

\includegraphics[trim={0.3cm 0.3cm 0.5cm 0cm}, clip, width=0.495\textwidth]{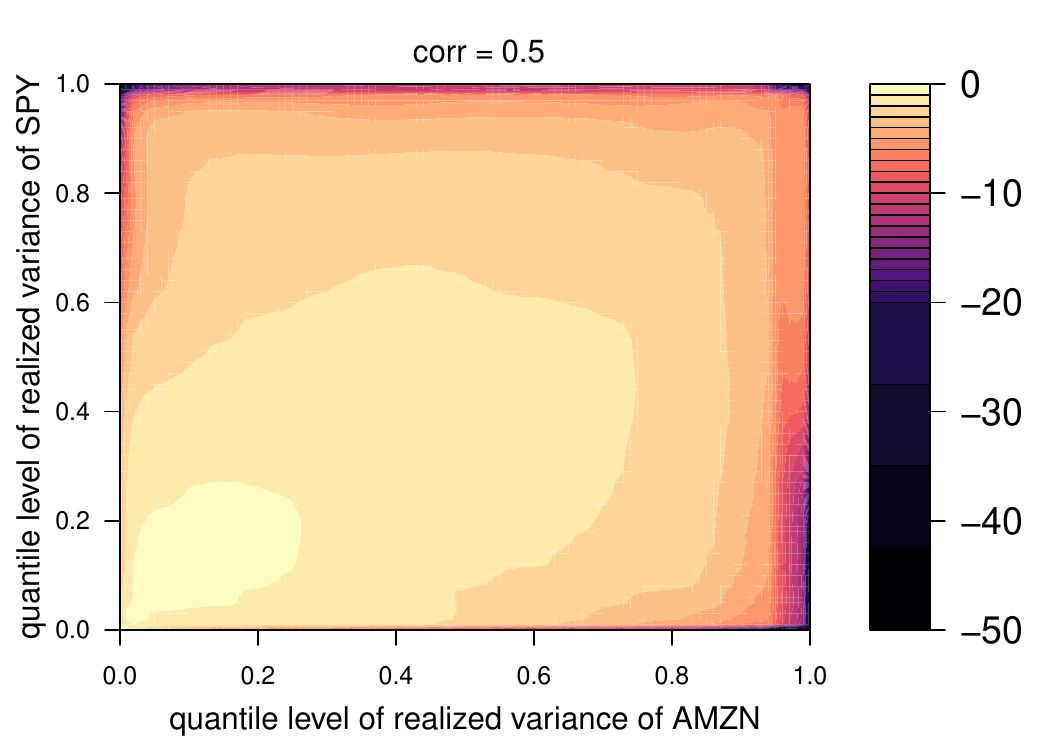}\hfill
\includegraphics[trim={0.3cm 0.3cm 0.5cm 0cm}, clip, width=0.495\textwidth]{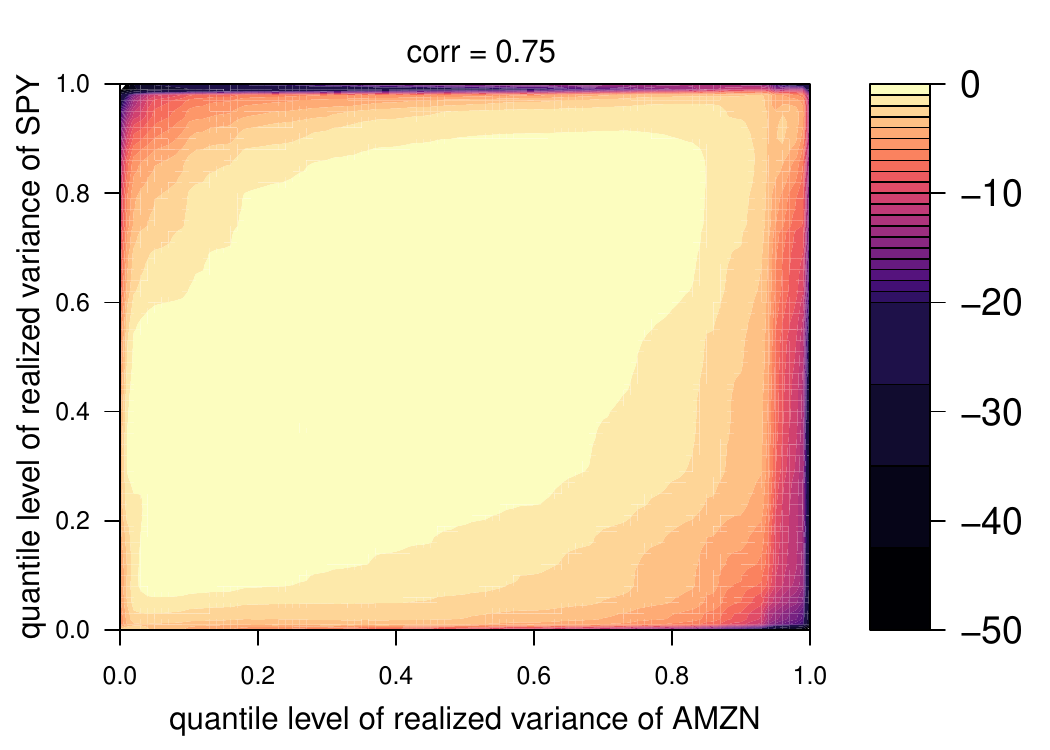}
\caption{Contour plots of $S\mapsto \smash{\log\{\hat{f}_{n,b_n^{\star}}^{\,\mathrm{W}}(S)\}}$ for cases of realized correlation of $0$, $0.25$, $0.5$, and $0.75$, rescaled so that the logarithm of the density equals zero at the maximum value.}
\label{fig:app1.contour}
\end{figure}

\section{List of abbreviations}\label{app:abbreviations}

\begin{tabular}{llll}
& a.s. &\hspace{20mm}& almost surely \\
& DTI &\hspace{20mm}& diffusion tensor imaging \\
& iid &\hspace{20mm}& independent and identically distributed \\
& IQR &\hspace{20mm}& interquartile range \\
& KDE &\hspace{20mm}& kernel density estimator/estimation \\
& LCV &\hspace{20mm}& likelihood cross-validation \\
& LSCV &\hspace{20mm}& least-squares cross-validation \\
& MAE &\hspace{20mm}& mean absolute error \\
& MSE &\hspace{20mm}& mean squared error \\
& RISE &\hspace{20mm}& root integrated squared error \\
& WAR &\hspace{20mm}& Wishart autoregressive (process) \\
\end{tabular}

\section{Reproducibility}\label{app:reproducibility}

The \textsf{R} package \texttt{ksm}, which is available from CRAN \citep{ksm}, implements the Wishart KDE and all functions needed to reproduce the simulation study and the real-data application presented herein. The code required to generate all figures and tables can be accessed through the GitHub repository at \href{https://github.com/FredericOuimetMcGill/WishartKDE}{https://github.com/FredericOuimetMcGill/WishartKDE}.

\end{appendices}

\section*{Acknowledgments}
\addcontentsline{toc}{section}{Acknowledgments}

Simulations were carried out using the computational resources supplied by Calcul Qu\'ebec and the Digital Research Alliance of Canada. The dataset featured in the real-data application in Appendix~\ref{app:application} was kindly provided by Prof.\ Anne MacKay (Universit\'e de Sherbrooke).

\section*{Funding}
\addcontentsline{toc}{section}{Funding}

Belzile acknowledges funding from the Natural Sciences and Engineering Research Council of Canada (NSERC) through Discovery Grant RGPIN-2022-05001. Genest's work was funded by NSERC Grant RGPIN-2024-04088 and the Canada Research Chairs Program (Grant 950-231937). Ouimet acknowledges funding from the NSERC through Discovery Grant RGPIN-2026-04471 and Discovery Launch Supplement DGECR-2026-00449. An early version of this work was completed while Ouimet was a Research Associate at McGill University and a postdoctoral fellow at the Université de Sherbrooke. These positions were funded through the research grants of Christian Genest and Anne MacKay, respectively.

\section*{References}
\addcontentsline{toc}{section}{References}

\setlength{\bibsep}{0pt plus 0ex}

\bibliographystyle{plainnat}
\bibliography{bib_clean}

\end{document}